\documentclass[article]{IEEEtran}

\usepackage[T1]{fontenc}
\usepackage{graphicx}
\usepackage{cite}
\usepackage{subcaption}

\usepackage{overpic}
\usepackage{amssymb}
\usepackage{amsmath}
\usepackage{amsthm}
\usepackage{amssymb}
\usepackage{steinmetz}
\usepackage{array}
\usepackage{lipsum}  
\usepackage{url}
\usepackage{mathrsfs}
\usepackage[normalem]{ulem}
\usepackage{cancel}

\usepackage[monochrome]{xcolor}

\DeclareMathOperator{\oj}{j}
\newcommand{\bsy}[1]{\boldsymbol{#1}}
\newcommand{\derpar}[2]{\dfrac{\partial {#1}} {\partial {#2}}}

\newcommand{\blue}[1]{\textcolor{blue}{#1}}

\theoremstyle{plain}

\newtheorem{corollary}{Corollary}
\newtheorem{remark}{Remark}
\newtheorem{proposition}{Proposition}
\newtheorem{assumption}{Assumption}

\newcommand{\vect}[1]{\mathbf{#1}}

\def\diag{\mathrm{diag}}

\def\imagunit{\mathsf{j}} % Imaginary number

\def\re{\mathrm{Re}}

\definecolor{bittersweet}{rgb}{1.0, 0.44, 0.37}

\usepackage{fontawesome}
\usepackage{calligra}

\DeclareMathAlphabet{\mathcalligra}{T1}{calligra}{m}{n}
\usepackage[mathscr]{eucal}

\begin{document}

\title{Cram{\'e}r-Rao Bounds for Holographic Positioning}%Using the Spherical Electric Field over Large Planar Surfaces}

\author{
\IEEEauthorblockN{Antonio A. D'Amico, Andrea de Jesus Torres, Luca Sanguinetti, \emph{Senior Member, IEEE},
Moe Win, \emph{Fellow, IEEE}\vspace{-0.7cm}
\thanks{
Part of this paper was presented at the 2021 Asilomar Conference on Signals, Systems, and Computers~\cite{torres2021cramerrao}.
\newline\indent A. A. D'Amico, A. de Jesus Torres and Luca Sanguinetti are with the Dipartimento di Ingegneria dell'Informazione, University of Pisa, Italy. (e-mail: \{a.damico, luca.sanguinetti\}@unipi.it, andrea.dejesustorres@phd.unipi.it).
\newline \indent M. Win is with the Wireless Information and Network Sciences Laboratory (WINSLab), Massachusetts Institute of Technology (MIT), Cambridge, MA 02139 USA (e-mail: moewin@mit.edu) 
\newline \indent The research was supported by the MIT-UNIPI grant (VIII call) from MISTI Global Seed Funds in the framework of the MIT-Italy Program. L. Sanguinetti and A. A. D'Amico were also partially supported by the Italian Ministry of Education and Research (MIUR) in the framework of the CrossLab project (Departments of Excellence).
}}
}
\maketitle

\begin{abstract}
Multiple antennas arrays \blue{combined with high carrier frequencies} play a key role in wireless networks for communications but also for localization and sensing applications. %The use of large antenna arrays at high carrier frequencies (in the mmWave range) pushes towards a propagation regime in which the wavefront is no longer plane but spherical. This allows to infer the position and orientation of a transmitting source from the received signal without the need of using multiple anchor nodes, located in known positions. 
To understand the fundamental limits of \blue{electromagnetically} large antenna arrays for localization, this paper combines wave propagation theory with estimation theory, and computes the Cram{\'e}r-Rao Bound (CRB) for the estimation of the source position on the basis of the three Cartesian components of the electric field, observed over a rectangular surface area. The problem is referred to as holographic positioning and \blue{it intrinsically depends on} the radiation angular pattern of the transmitting source, which is typically ignored in standard signal processing models. We assume that the source is a Hertzian dipole, and address the holographic positioning problem in both cases, that is, with and without a priori knowledge of its orientation. To simplify the analysis and gain further insights, we also consider the case in which the dipole is located on the line perpendicular to the surface center. Numerical and asymptotic results are given to quantify the CRBs, and to quantify the effect of various system parameters on the ultimate estimation accuracy. It turns out that surfaces of size, comparable to the distance, are needed to guarantee a centimeter-level accuracy in the mmWave bands. \blue{Moreover, we show that the CRBs with and without a priori knowledge of the source dipole orientation are numerically the same. The provided CRBs are also used to benchmark different maximum-likelihood estimators (MLEs) derived on the basis of a discrete representation of different models of the electric field. The analysis shows that, if the standard models are used (neglecting the radiation angular pattern), the MLE accuracy is far from the CRB. On the other hand, it approaches the CRB when the more detailed electromagnetic model is considered.}

\end{abstract}
\smallskip
\begin{IEEEkeywords}
Cram{\'e}r-Rao bound, wave propagation theory, source localization, planar electromagnetic surfaces, holographic positioning.
\end{IEEEkeywords}

% Introduction
 %!TEX root = main.tex
\section{Introduction}
%Location-aided systems are expected to have a wide range of applications in communications systems, whether for vehicular communications, assisted living applications, or to support the communication robustness and effectiveness in different aspects such as resource allocation, beamforming, and pilot assignment. While the positioning requirements were traditionally determined by mandates to localize emergency calls, every new generation of cellular networks has provided new positioning opportunities. To break the meter accuracy barrier and support new applications, disruptive solutions are needed.

The estimation accuracy of signal processing algorithms for positioning is fundamentally limited by the quality of the underlying measurements. For time-based measurements, high resolution and high accuracy can only be obtained when a large bandwidth is available. Improvements can be achieved by using multiple anchor nodes, located in known positions. Antenna arrays have thus far only played a marginal role in positioning since the small arrays of today's networks provide little benefit. With future networks, the situation may change significantly. Indeed, the 5G technology standard\textcolor{blue}{~\cite{3gppRel16}} is envisioned to operate in the mmWave bands~\cite{Lee2018Spectrum5G}, while 6G research is already focusing on the so-called sub-terahertz (THz) bands, i.e., in the range 100 -- 300\,GHz. The small wavelength of high-frequency signals makes it practically possible to envision arrays with a very large number of antennas, as never seen before. The advent of large spatially-continuous electromagnetic surfaces interacting with wireless signals pushes even further this vision. \textcolor{blue}{From the technological point of view, metamaterials represent appealing candidates toward the creation of software-controlled metasurfaces, which can lead to a viable way of realizing spatially-continuous electromagnetic surfaces~\cite{dardari2020holographic}.} Research in this direction is taking place under the names of Holographic MIMO~\cite{Huang2020,Pizzo2019a}, large intelligent surfaces~\cite{Hu2018a,Hu2018b}, and reconfigurable  intelligent surfaces~\cite{Basar2019a,DiRenzo2020}. All this opens new dimensions and brings new opportunities for communications but also for localization and sensing~\cite{bourdoux20206g}. 
\vspace{-0.3cm}
\subsection{Motivation and contribution}

A side-effect of using large arrays or surfaces combined with high carrier frequencies is to push the electromagnetic propagation  towards the regime in which the wavefront associated to the signal transmitted by the source tends to be spherical and cannot be approximated by a plane wave. In this regime, also the distance information, not only the angle-of-arrival, can be inferred from it. \textcolor{blue}{This concept is not new~\cite{Benjamin2019} and it has been widely used to develop signal processing algorithms that exploit the spherical wavefront properties to communicate in low rank propagation channels (e.g.,~\cite{Ingram2005,Bohagen2009,JinLensArray2020}) and to pinpoint the position of the source with high accuracy~\cite{Hu2018b,AlegriaLISPositioning2019,Dardari2021,DELMAS2016117,Cadre1995,Grosicki2005,Korso2012,Siwei2019,Xuefeng2017,Guerra2021}.} In this latter context, the question arises of the ultimate accuracy that can be achieved in localization operations. This is important in order to provide benchmarks for evaluating the performance of actual estimators. 

Motivated by the above consideration, in this paper we combine electromagnetic propagation concepts with estimation theory, and compute the Cram{\'e}r-Rao Bound (CRB) for source localization. \textcolor{blue}{In doing so, we consider the three Cartesian components of the electric vector field, observed over a rectangular surface, situated in the Fraunhofer radiation region of the source~\cite[Ch. 15]{OrfanidisBook}. In general, the three Cartesian components depend on the radiation vector~\cite[Ch. 15]{OrfanidisBook}, which is in turn determined by the current distribution inside the source. This functional dependence is typically overlooked in \textit{standard} signal processing models used for CRB computations and may lead to estimation algorithms with less accuracy than those based on the complete model.} %Despite possible, it must be clear that such approach has the following relevant shortcomings. Firstly, the CRBs computed by using approximate models that neglect relevant information about the source position are inevitably \emph{unaccurate} to quantify the ultimate accuracy in source localization. Secondly, the true model may depend on unknown parameters (i.e., \textit{nuisance parameters}) that, although not related to the source position, make the CRB obtained by ignoring them wrong. 
\blue{Note that similar considerations can be found in~\cite{Benjamin2019} where \textit{standard} models are compared to the so-called \textit{analytic model}, based on electromagnetic theory. Particularly,~\cite{Benjamin2019} observes that the former typically ignore the nature of the source whose physical characteristics have a profound impact on the generated electromagnetic fields. We will elaborate further on the differences between what we call analytic model and what we collectively denote as standard models later on, after introducing the necessary mathematical machinery.}

% \blue{We will elaborate further on this throughout the paperA discussion about the differences between what we call analytic model and what we collectively denote as standard models is postponed until Section III.C after introducing the necessary mathematical machinery.}

\blue{The first objective of this paper is to compute and analyze the CRB for the localization of a source on the basis of the analytic model, which stems from first principles of electromagnetic theory. In doing so, we assume that the source is an elementary Hertzian dipole and make use of all the three cartesian components of the electric field for estimating its position. This allows us to derive a more fundamental limit to the accuracy of estimators that may possibly exploit the entire electric field. This is what we call \emph{holographic positioning} where the holographic term dates back to the ancient greek and literally means ``describe everything''~\cite{dardari2020holographic}. {\color{orange}The concept is often connected with metasurfaces, which are two-dimensional surfaces consisting of arrays of reconfigurable elements of metamaterial. We refer to~\cite{Tsilipakos2020} for a recent survey on the implementation aspects for metasurfaces}. The main results of our CRB analysis can be summarized as follows. 
\begin{itemize}
  \item We show how the ultimate estimation accuracy for the localization of a source depends strongly on its orientation. To the best of authors' knowledge, such a dependence has never been pointed out before, and comes from considering the radiation vector in the expression of the received electric field.
  \item The orientation of the source can be assumed known or unknown to the receiver. A second interesting contribution of this paper is to show that, under practical conditions, the CRB computed assuming that the orientation is unknown coincides with that computed assuming that it is perfectly known. This does not mean that we can ignore the effect of dipole orientation in the estimation process but only that the joint estimation of orientation and position ultimately provides the same localization accuracy as if the orientation were known.
  \item To gain insights about the impact on the estimation accuracy of different system parameters (such as wavelength, size of the receiving surface), we assume that the dipole is located on the line perpendicular to the surface center. In this case, closed-form expressions can be computed that show that the CRB scales quadratically with the wavelength. This is in line with the results in \cite{Hu2018b} where a standard model was used. Also, we show that the accuracy in the estimation of some components of the position vector improves unboundedly as the ratio between the size of the receiving surface and its distance from the source increases. This is a new result that cannot be found in prior works. 
\end{itemize}
}
%Our objective here is to provide the analytic signal model based on first principles in electromagnetic theory and use it to compute the CRBs for the localization of a Hertzian dipole, that has been
%largely considered in the literature but using \textit{standard} models. Further simplifications and insights are obtained by assuming that the dipole is located on the line perpendicular to the surface center. To make the treatment more general, we consider a vector wave model. Note that scalar approximations are often used to simplify the analysis and provide an intermediate step toward the understanding of the vector wave model. These simplifications may simply follow from considering one of the three components of the electric field or also from using the component of the Poynting vector perpendicular (in each point) to the observation region~\cite{{Hu2018b}}. Not relying on the above assumptions allows us to derive a more fundamental limit to the accuracy of actual estimators that may possibly exploit the entire electric field. This is called \emph{holographic positioning} where the holographic term dates back to the ancient greek and literally means ``describe everything''~\cite{dardari2020holographic}.
\blue{A second important objective of this paper is to use the CRB analysis to understand whether or not positioning algorithms based on the analytic model can provide gains, compared to those based on standard models. To partially answer this question, we consider a planar array made of Hertzian dipoles (filling the receiving surface) and we analyze the performance of the three different maximum-likelihood estimators (MLEs) of the source position, which are derived on the basis of three different models of the electric field. In particular, the first estimator makes use of the analytic model adopted in this paper, the second one is based on the model from~\cite{Hu2018b}, and the third one relies on the common planar approximation of the received electromagnetic wave.  Numerical results are used to compare the different estimators. It turns out that the estimation accuracy of the MLE based on the analytic model is very close to the theoretical limits provided by the CRB. Moreover, we show that it outperforms the other two estimators as it takes into account the orientation of the transmitting dipole. Finally, we show that the estimation accuracy depends on the source orientation, but not on its a priori knowledge. This is exactly in line with the behavior predicted by the CRB analysis.}

%\textcolor{red}{Our objective here is to understand whether or not positioning algorithms based on a more accurate expressions of the electromagnetic fields can provide some gain compared to localization techniques based on prior adopted models, and, in case, under what conditions. An answer to these questions can be found in the CRB expressions provided in our paper. how far the performance of actual estimators are from the limits provided by the CRB ,.... and how the electromagnetic model }

Compared to the conference version~\cite{torres2021cramerrao}, this paper provides a more detailed derivation and discussion of the electromagnetic vector model and computes the CRBs in both cases, that is, with and without a priori knowledge of the source dipole orientation. Also, it provides all the derivations that were omitted in~\cite{torres2021cramerrao}, \blue{and analyzes the performance of MLEs based on different electric field models.}

\subsection{Paper outline and notation}
The paper is organized as follows. In Section~\ref{SectionII}, we start from first principles of wave propagation theory and provide the most general form of the electric field, which is then used to compute the CRB for the estimation of the position of the transmitting source. We also briefly describe how simplified models can be obtained from the general form. In Section~\ref{sec:est3D}, we simplify the analysis by assuming that the transmitting source is an elementary dipole pointing in an arbitrary direction. The CRBs for the estimation of the dipole position are computed for both cases, i.e., with and without a priori knowledge of its orientation. To gain insights into the CRBs, in Section~\ref{Sec:case-study} we assume that the dipole is located in the central perpendicular line (CPL) and is vertically oriented. Insightful closed-form expressions are provided to quantify the ultimate estimation accuracy. In Section~\ref{Sec:Numerical_analysis}, we provide some numerical results as a function of system parameters, e.g., distance, array size, carrier frequency. \blue{To quantify the impact of the electric field model on the estimation accuracy, in Section VI we analyze the performance of three different MLEs. The first is based on the adopted analytic model while the other two are based on models adopted in the literature.} The major conclusions are drawn in Section~\ref{Sec:conclusions}.

The following notation is used throughout the paper. In the $\mathbb{R}^3$ Euclidean space, an arbitrary spatial vector ${\bf r}$  is represented as $\vect{r} =(\alpha,\beta,\gamma)$, where $(\alpha,\beta,\gamma)$ are the \textit{components} of $\vect{r}$ along the directions of three given orthonormal vectors ${\bf \hat u}_1, {\bf \hat u}_2, {\bf \hat u}_3$. Equivalently, we write $\vect{r}=\alpha {\bf \hat u}_1+\beta {\bf \hat u}_2+\gamma {\bf \hat u}_3$.
The length of $\vect{r}$ is $||{\bf r}||= \sqrt{\alpha^2 + \beta^2 + \gamma^2}$, and $\bf{\hat r}= \vect{r}/||{\bf r}||$ is a unit vector that is parallel to $\vect{r}$. We use $ {\bf a} \times {\bf b}$ and ${\bf a}\cdot{\bf b}$ to denote, respectively, the cross and the dot product between ${\bf a}$ and ${\bf b}$. An arbitrary point $P$ in $\mathbb{R}^3$ is described by its three \textit{coordinates} with respect to a reference system (cartesian or spherical). A point $P$ can also be represented by a vector ${\bf r}$. In such a case, $P$ is the endpoint of ${\bf r}$, whose starting point is fixed. 

\subsection{Reproducible research}
The Matlab code used to obtain the simulation results will be made available upon completion of the review process.

% system model
 %!TEX root = main.tex

%	\caption{$\mathbf{r}$ is a generic point in the space, while $\mathbf{O}$ is and $\mathbf{r}'$ are points belonging to the volume $\mathbf{V}$. By taking $\mathbf{O}$ as a reference point, whose distance from $\mathbf{r}$ is equal to $r$, we can express the distance between $\mathbf{r}$ and $\mathbf{r}'$ as $R = |\mathbf{r}-\mathbf{r}'|$. When $|\mathbf{r}|$ is much larger than the the greatest dimension of $\mathbf{V}$ we can approximate $R$ as $R \approx r- \hat{\mathbf{r}}\cdot \mathbf{r}'$} 

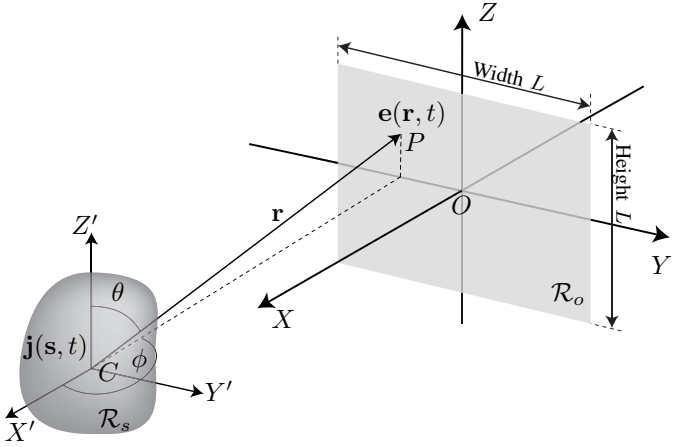
\begin{figure}
	
	\centering
	\begin{overpic}[width=1\columnwidth]{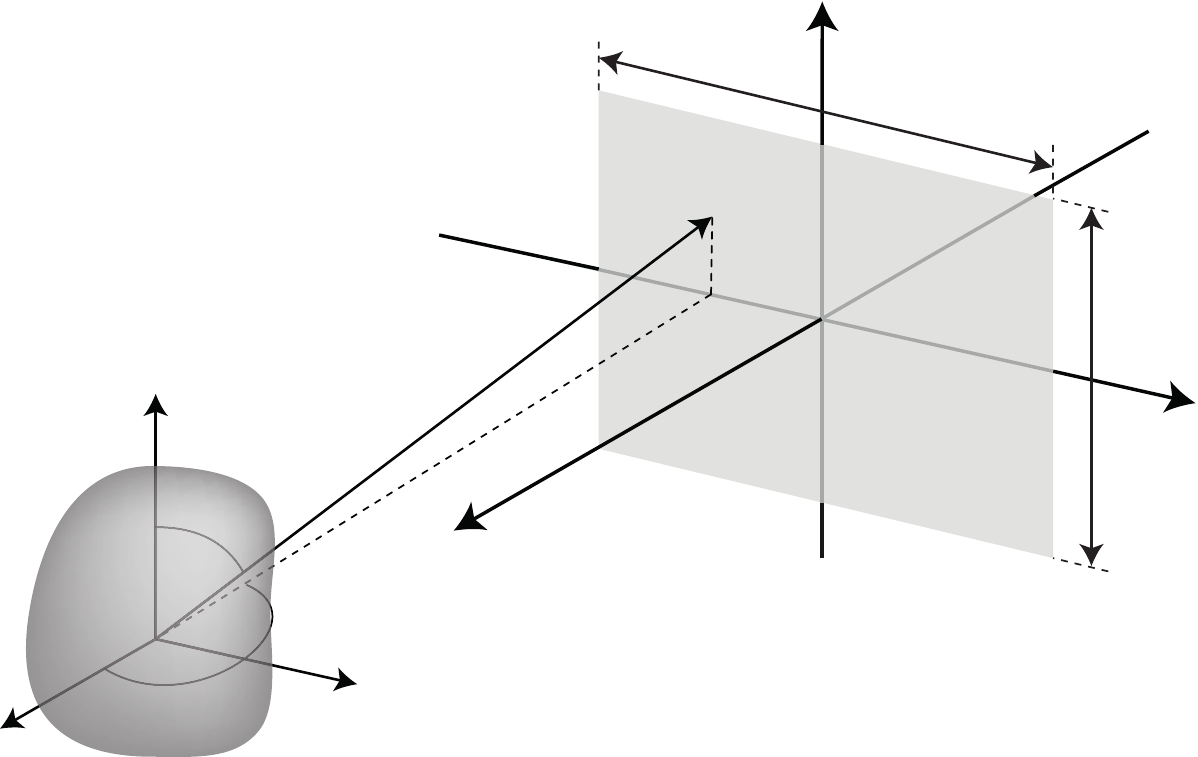}
	\put(67,33){$O$} %52 40
	\put(40,16){$X$}
	\put(97,24){$Y$}
	\put(71,62){$Z$}
	
	\put(14,8){$C$}
	\put(0,-1){$ X'$}
	\put(30,5){$Y'$}
	\put(10,30){$Z'$}
%	\put(2,0){$\widehat{\bf  x}'$}
%	\put(30,5){$\widehat{\bf  y}'$}
%	\put(10,30){$\widehat{\bf  z}'$}

	\put(56,47){${\mathbf{e}}({\mathbf{r}},t)$}
	\put(60,42.5){$P$}
%	\put(28,46){$\color{red}{e_x}$}
%	\put(41.5,48){$\color{teal}{e_y}$}
%	\put(30.5,54){$\color{blue}{e_z}$}
	
	\put(3,12){${\mathbf{j}}({\mathbf{s}},t)$}

	\put(14,1){$\mathcal R_s$}
	\put(82,20){$\mathcal R_o$}
	
	\put(40,32){$\mathbf{r}$}
	\put(16,20){$\theta$}
	\put(19,10){$\phi$}
	
	\put(70,54){\rotatebox{-13}{\footnotesize{Width $L$}}}
	
        \put(91.1,43.6){\rotatebox{-90}{\footnotesize{Height $L$}}}
        \end{overpic}
       	
        \caption{Geometry of the considered system.}\vspace{-0.5cm}
        \label{fig:source_volume}

\end{figure}
\section{Signal model and problem formulation}\label{SectionII}

Consider the system depicted in Fig.~\ref{fig:source_volume} in which an electric current density ${\bf j}({\bf s}, t)$, inside a source region $\mathcal R_s$, \textcolor{orange}{with ${\bf s}$ a spatial vector identifing a generic point in $\mathcal R_s$,} generates an electric field ${\bf e}({\bf r}, t)$ at a generic point $P$, identified through the spatial vector $\vect{r}$.
%NOTA BENE: A questo stadio, dire che $P$ \`e individuato tramite $\vect{r}$ serve solo per giustificare la notazione ${\bf e}({\bf r}, t)$}.
We consider only monochromatic sources and fields of the form 
\begin{equation}
{\bf j}(\vect{s},t) = \mathrm{Re}\left\{{\bf j}(\vect{s}) \mathrm{e}^{\imagunit \omega t}\right\}
\end{equation}
 and 
 \begin{equation}
 {\bf e}(\vect{r}, t) = \mathrm{Re}\left\{{\bf e}(\vect{r}) \mathrm{e}^{\imagunit \omega t}\right\}\end{equation} 
 where $\omega$ is frequency in radians/second. In this case, Maxwell's equations can be written only in terms of the current and field \textit{phasors}, ${\bf j}(\vect{s})$ and ${\bf e}(\vect{r})$ \cite[Ch. 1]{ChewBook}.

We call $C$ the \textit{centroid} of the source region $\mathcal R_s$ and assume that the electric field ${\bf e}(\vect{r})$, produced by  ${\bf j}(\vect{s})$, is measured over an \textit{observation region} $\mathcal R_o$ outside $\mathcal R_s$. The electromagnetic field propagates in a homogeneous and isotropic medium with neither obstacles nor reflecting surfaces. In other words, there is only a line-of-sight link between $\mathcal R_s$ and $\mathcal R_o$. %We make the following assumptions.

\subsection{Signal model}
The measured field is the sum of ${\bf e}(\vect{r})$ and a random noise field ${\bf n}(\vect{r})$, i.e.,
\begin{equation}
\label{SigNoise}
{\boldsymbol \xi} (\vect{r})={\bf e}(\vect{r})+{\bf n}(\vect{r})
\end{equation}
where ${\bf n}(\vect{r})$ is generated by electromagnetic sources outside $\mathcal R_s$. \textcolor{blue}{In an unbounded, homogeneous and isotropic medium, the electric field ${\bf e}(\vect{r})$ can be written as \cite[Ch. 1]{ChewBook}
\begin{equation}
\label{e1}
{\bf e} ({\bf r})=-\imagunit k Z_0 \int\limits_{\mathcal R_s}  \overline{\overline{\bf G}}(\mathbf{r}-\mathbf{s}) \cdot {\bf j}(\mathbf{s})  {\rm d}\mathbf{s}
\end{equation}
where $k=2 \pi/ \lambda$ is the \textit{wavenumber}, $\lambda=2 \pi c /\omega$ is the \textit{wavelength}, $Z_0$ is the intrinsic impedance of the medium, and $\overline{\overline{\bf G}}(\mathbf{p})$ is the dyadic Green's function, given by \cite{poon2005degrees}
\begin{equation}
\label{DyadicGF1}
\begin{split}
\overline{\overline{\bf G}}(\mathbf{p})&= g(p)\left[\left(1-\imagunit\dfrac{1}{kp}-\dfrac{1}{k^2p^2}\right)\overline{\overline{\bf I}} -\left(1-\imagunit \dfrac{3}{kp} - \dfrac{3}{k^2p^2}\right) \hat{\bf{p}} \hat{\bf{p}}\right]
\end{split}
\end{equation}
where $\overline{\overline{\bf I}}$ is the unit dyad, $p=\|{\bf p}\|$, $\hat{\bf{p}}={\bf p}/p$, and $g(p)$ is the scalar Green's function, i.e., 
\begin{equation}
\label{GreenScalare}
g(p)=\dfrac{e^{- \imagunit k p}}{4 \pi p}.
\end{equation}
Consider a cartesian coordinate system $CX'Y'Z'$ with the origin in the centroid $C$ in Fig.~\ref{fig:source_volume}, and let ${\bf{\hat x}}$, ${\bf{\hat y}}$ and ${\bf{\hat z}}$, the unit vectors in the $X'$, $Y'$ and $Z'$ directions, respectively. We make the following assumption.
\begin{assumption}[\textcolor{blue}{Fraunhofer radiation region of the source}]
Let $r_o$ be the minimum distance of $C$ from $\mathcal R_o$ and denote by $l_s$ the largest dimension of $\mathcal R_s$. We assume that $r_o \gg l_s$ and $r_o \gg 2 l_s^2 /\lambda$. These conditions define the so-called far-field or Fraunhofer radiation region of the source~\cite[Ch. 15]{OrfanidisBook}.
\label{ass:Frau}
\end{assumption}
\textcolor{blue}{Under Assumption 1}, the electric field ${\bf e}(\vect{r})$ can be approximated as~\cite[Ch.\,15]{OrfanidisBook}:
\begin{equation}
\label{E_P.0}
{\bf e}(\vect{r})= G(r) \left[{\bf \hat r} {\bf \times} {\bf R}(\theta,\phi) \right] {\bf \times} {\bf \hat r}
\end{equation}
where 
\begin{equation}
\label{GreenScalare2}
G(r)=-\imagunit k Z_{0} g(r)
\end{equation}
and $(r,\theta,\phi)$ are the spherical coordinates (with respect to $CX'Y'Z'$) of a generic point ${\bf r} \in \mathcal R_o$, i.e., ${\bf r} = r \cos \phi \sin \theta {\bf \hat x} + r \sin \phi \sin \theta {\bf \hat y} + r \cos \theta {\bf \hat z}$. Also, ${\bf \hat r}$ is the unit vector in the radial direction and ${\bf R}(\theta,\phi)$ is the \textit{radiation vector}. This is related to the source current distribution ${\bf j}({\bf s})$ as follows~\cite[Eq. (15.7.5)]{OrfanidisBook}
\begin{equation}
\label{RadVect}
{\bf R}(\theta,\phi)=\int_{\mathcal R_s} {\bf j}({\bf s}) \mathrm{e}^{\imagunit {\bf k}(\theta,\phi) \cdot {\bf s}} d {\bf s}
\end{equation}
where ${\bf k}(\theta,\phi)=k {\bf \hat r}$ is the \textit{wavenumber vector}. Denote by $R_{r}(\theta,\phi)$, $R_{\theta}(\theta,\phi)$ and $R_{\phi}(\theta,\phi)$ the components of the radiation vector ${\bf R}(\theta,\phi)$ along the ${\bf \hat r}$, ${\bsy{\hat{\theta}}}$ and ${\bsy{\hat{\phi}}}$ directions. Then, we may write:
\begin{equation}
\label{RadVect.1}
{\bf R}(\theta,\phi)=R_{r}(\theta,\phi){\bf \hat r}+ R_{\theta}(\theta,\phi) {\bsy{\hat{\theta}}}+R_{\phi}(\theta,\phi) {\bsy{\hat{\phi}}}
\end{equation}
where 
\begin{align}\label{r}
{\bf \hat r} & = \sin \theta \cos \phi {\bf \hat x} + \sin \theta \sin \phi{\bf \hat y} + \cos \theta{\bf \hat z}\\\label{theta}
{\bsy{\hat{\theta}}} & =  \cos \theta \cos \phi {\bf \hat x} +  \cos \theta \sin \phi{\bf \hat y} - \sin \theta {\bf \hat z}\\\label{phi}
{\bsy{\hat{\phi}}} & = - \sin \phi{\bf \hat x} + \cos \phi{\bf \hat y}.
\end{align}
Plugging \eqref{RadVect.1} into \eqref{E_P.0} yields
\begin{equation}
\label{E_P}
{\bf e}(\vect{r})= G(r) \left[R_{\theta}(\theta,\phi) {\bsy{\hat{\theta}}}+R_{\phi}(\theta,\phi) {\bsy{\hat{\phi}}} \right]
\end{equation}
since ${\bf \hat r} \times {\bf \hat r} = {\bf 0}$, $({\bf \hat r} \times {\bsy{\hat{\theta}}} )\times {\bf \hat r} = {\bsy{\hat{\theta}}}$ and $({\bf \hat r} \times {\bsy{\hat{\phi}}} )\times {\bf \hat r} = {\bsy{\hat{\phi}}}$.\footnote{Notice that ${\bf \hat r} \times {\bsy{\hat{\theta}}} = -{\bsy{\hat{\phi}}}$, $-{\bsy{\hat{\phi}}}\times {\bf \hat r}={\bsy{\hat{\theta}}}$, ${\bf \hat r} \times {\bsy{\hat{\phi}}} = -{\bsy{\hat{\theta}}}$ and $-{\bsy{\hat{\theta}}}\times {\bf \hat r}={\bsy{\hat{\phi}}}$.} Notice that ${\bf e}(\vect{r})$ in~\eqref{E_P} is completely determined by the transverse component 
\begin{equation}\label{eq:transverse_component}
{\bf R}_{\perp}(\theta,\phi) = R_{\theta}(\theta,\phi) {\bsy{\hat{\theta}}}+R_{\phi}(\theta,\phi) {\bsy{\hat{\phi}}} 
\end{equation}
of the radiation vector ${\bf R}(\theta,\phi)$~\cite[Ch.~15]{OrfanidisBook}. Similar to \cite{Benjamin2019}, in what follows we refer to \eqref{E_P.0} as the \textit{analytic} model.\\
\indent Denote by $\xi_x(\vect{r})$, $\xi_y(\vect{r})$ and $\xi_z(\vect{r})$, the cartesian components of ${\boldsymbol \xi} (\vect{r})$ along the ${\bf \hat x}$, ${\bf \hat y}$ and ${\bf \hat z}$ directions, respectively. From \eqref{SigNoise}, we have
\begin{equation}
\label{SigNoiseX}
\xi_x(\vect{r})= [{\bf e}(\vect{r})+{\bf n}(\vect{r})]\cdot {\bf \hat x}=
e_x(\vect{r})+n_x(\vect{r})
\end{equation}
\begin{equation}
\label{SigNoiseY}
\xi_y(\vect{r})= [{\bf e}(\vect{r})+{\bf n}(\vect{r})]\cdot {\bf \hat y} = e_y(\vect{r})+n_y(\vect{r})
\end{equation}
\begin{equation}
\label{SigNoiseZ}
\xi_z(\vect{r})= [{\bf e}(\vect{r})+{\bf n}(\vect{r})]\cdot {\bf \hat z}= e_z(\vect{r})+n_z(\vect{r})
\end{equation}
%where
%\begin{equation}
%\label{ }
%\nonumber
%e_x(\vect{r})={\bf e}(\vect{r}) \cdot \hat{\bf x} \quad  \quad  e_y(\vect{r})={\bf e}(\vect{r}) \cdot \hat{\bf y}\quad  \quad  e_z(\vect{r})={\bf e}(\vect{r}) \cdot \hat{\bf z}
%\end{equation}
where $e_x(\vect{r}),e_y(\vect{r})$ and $e_z(\vect{r})$ are obtained from \eqref{E_P} by using~\eqref{r} -- \eqref{phi}. This yields
\begin{align}
\label{ex}
e_x(\vect{r})&=G(r) \left[R_{\theta}(\theta,\phi) \cos \theta \cos \phi-R_{\phi}(\theta,\phi) \sin \phi \right]\\
\label{ey}
e_y(\vect{r})&=G(r) \left[R_{\theta}(\theta,\phi) \cos \theta \sin \phi+R_{\phi}(\theta,\phi) \cos \phi \right]\\
\label{ez}
e_z(\vect{r})&=-G(r) R_{\theta}(\theta,\phi) \sin \theta.
\end{align}
\subsection{Problem formulation}
\label{CRB_full}
We aim at computing the CRB for the estimation of the position of the centroid $C$ in Fig.~\ref{fig:source_volume} based on the noisy vector ${\boldsymbol \xi} (\vect{r})$ over the observation region $\mathcal R_o$, whose cartesian components are given by~\eqref{SigNoiseX} -- \eqref{SigNoiseZ}. For this purpose, the following assumptions are further made.
\begin{assumption}\label{assumption1}
The observation region is a square domain parallel to the $Y'Z'$ coordinate plane. In particular, assume that
\begin{equation}\notag
\mathcal R_o=\big\{ (x',y',z'): x'=x'_o, |y'-y'_o| \le L/2, |z'-z'_o| \le L/2 \big\}
\end{equation}
where $(x'_o,y'_o,z'_o)$  are the cartesian coordinates of the center $O$ of $\mathcal R_o$ in the system $CX'Y'Z'$. 
\end{assumption}
\begin{assumption}[\textcolor{blue}{Random noise field modelling}]
Following~\cite{Jensen2008channel}--\hspace{1sp}\cite{Gruber2008EIT}, we model ${\bf n}(\vect{r})$ as a spatially uncorrelated zero-mean complex Gaussian process with correlation function 
\begin{equation}
\label{ }
\mathrm{E} \left\{{\bf n}(\vect{r}) {\bf n}^{\dagger}(\vect{r}') \right\}=\sigma^2 {\bf I} \delta(\vect{r}-\vect{r}')
\end{equation}
where ${\bf I}$ is the identity matrix, $\delta(\cdot)$ is the Dirac's delta function, and $\sigma^2$ is measured in $\mathrm V ^2$, where $\mathrm V$ indicates volts~\cite{Gruber2008EIT}.
\end{assumption}
The cartesian system $OXYZ$  in Fig.~\ref{fig:source_volume} is obtained by $CX'Y'Z'$ through a pure translation. The position of $C$ in the system $OXYZ$ is given by the cartesian coordinates $(x_C,y_C,z_C)$. Accordingly, we have that $x'=x-x_C$, $y'=y-y_C$ and $z'=z-z_C$.}

Let ${\bf u}=(x_C,y_C,z_C)$ denote the vector collecting the \emph{unknown} coordinates of $C$. The CRB for the estimation of the $i$th entry of ${\bf u}$, say $u_i$, is (e.g.,~\cite{Kay1993a})
\begin{equation}
\label{eq:CRBs}	
\mathrm{CRB}(u_i) = \left[\mathbf{F}^{-1}\right]_{ii}
\end{equation}
where ${\bf F}$ is the Fisher's Information Matrix (FIM). The latter is a $3\times 3$ hermitian matrix, whose elements are computed as~\cite[Appendix 15C]{{Kay1993a}}
\begin{equation}
\label{FIM}
\begin{split}
[\mathbf{F}]_{mn}=\dfrac{2}{\sigma^2}\re\left\{\iint_{-\frac{L}{2}}^{\frac{L}{2}} f_{mn}(y,z) dy dz\right\}
\end{split}
\end{equation}
where 
\begin{equation}
\label{fmn}
f_{mn}(y,z)=\derpar{e_x(\vect{r})}{u_m} \derpar{e_x^\ast(\vect{r})}{u_n}+\derpar{e_y(\vect{r})}{u_m} \derpar{e_y^\ast(\vect{r})}{u_n} +\derpar{e_z(\vect{r})}{u_m} \derpar{e_z^\ast(\vect{r})}{u_n}
\end{equation}
and the integration is performed over the observation region $\mathcal R_o$, as defined in Assumption~\ref{assumption1}. 
%For notational simplicity, in \eqref{fmn} the dependence of $e_x$, $e_y$ and $e_z$ on $\bf r$ has been omitted. 
The functional dependence of $e_x(\vect{r})$, $e_y(\vect{r})$ and $e_z(\vect{r})$ on ${\bf u}$ is hidden in the spherical coordinates $(r,\theta,\phi)$. Indeed, we have
\begin{align}
\label{SferCart1}
r &=||{\bf r}||=\sqrt{x^2_C+(y-y_C)^2+(z-z_C)^2}\\
\label{SferCart2}
\cos \theta &=\dfrac{z-z_C}{r}\\
\label{SferCart3}
\tan \phi &=-\dfrac{y-y_C}{x_C}.
\end{align} 
Accordingly, we can write
\begin{equation}
\label{DerparE}
\derpar{e_v(\vect{r})}{u_i}=\derpar{e_v(\vect{r})}{r}\derpar{r}{u_i}+\derpar{e_v(\vect{r})}{\theta}\derpar{\theta}{u_i}+\derpar{e_v(\vect{r})}{\phi}\derpar{\phi}{u_i}
\end{equation}
with $v \in \{x,y,z\}$ and $i \in \{1,2,3\}$. The derivatives of $r$, $\theta$ and $\phi$ with respect to $u_i$ can be easily computed from \eqref{SferCart1} -- \eqref{SferCart3}, whereas the partial derivatives of $e_v(\vect{r})$ with respect to $r$, $\theta$ and $\phi$ can be obtained from \eqref{ex} -- \eqref{ez}.  

It is clear that the computation of \eqref{DerparE} requires knowledge of the components $R_{\theta}(\theta,\phi)$ and $R_{\phi}(\theta,\phi)$ of the radiation vector ${\bf R}(\theta,\phi)$. In other words, the evaluation of \eqref{DerparE} implicitly assumes that current distribution inside the source region $\mathcal R_s$ is a priori known. However, this is not always the case in practical applications. When it is \emph{not} a priori known, the problem must be reformulated to take into account this lack of information. In the remainder of this paper, we consider a specific scenario in which the current source is an elementary (i.e., short) dipole with an arbitrary orientation and address the estimation problem in both cases, that is, \emph{with} and \emph{without} a priori knowledge of the current distribution.

\vspace{-0.7cm}
\textcolor{blue}{\subsection{Discussion on prior adopted electromagnetic models}}

\textcolor{blue}{The expression in \eqref{E_P} provides the electric field ${\bf e}(\vect{r})$ in a point ${\bf r}$ lying in the Fraunhofer radiation region of the source (i.e., under Assumption~1). We see that ${\bf e}(\vect{r})$ is the product of the two terms: $G(r)$ and ${\bf R}_{\perp}(\theta,\phi) = R_{\theta}(\theta,\phi) {\bsy{\hat{\theta}}}+R_{\phi}(\theta,\phi) {\bsy{\hat{\phi}}}$. The first term is given in \eqref{GreenScalare2} and represents a scalar spherical wave, which accounts for the distance $r$ between the source and the point, as given in~\eqref{SferCart1}. The second term is in~\eqref{eq:transverse_component} and takes into account the vector nature of ${\bf e}(\vect{r})$ as well as its \emph{dependence on the current distribution inside the transmitting volume}. In general, such a dependence is overlooked in the literature. To the best of our knowledge,~\cite{Benjamin2019} is the only paper in which it is recognized that the commonly adopted models do not consider the characteristics of the source (transmit antenna type, size, orientation, etc.), although it may significantly affect the structure of the received electric field.} 
\blue{In general, the standard models for CRB computation and position estimation have the following form~\cite{Benjamin2019}:
\begin{equation}
\label{SM1}
 \eta({\bf r}) = s({\bf r}) + n({\bf r}) 
 \end{equation}
where 
\begin{equation}
\label{SM2}
 s({\bf r})= f({\bf e}(\vect{r}))
 \end{equation}
 is some \emph{scalar} function $f(\cdot)$ of the received electric field and $n({\bf r})$ is additive noise.}
\textcolor{blue}{The vast majority of models is obtained as a further simplification of~\eqref{SM1} in which only the spherical wave term $G(r)$ is considered, e.g.,~\cite{Guerra2021}. This leads to:
\begin{equation}
\label{SM3}
 s({\bf r}) = \alpha G(r)
\end{equation}
where $\alpha$ is a scaling parameter independent of the observation point.
The analyses provided in \cite{Dardari2021}, \cite{DELMAS2016117}, and \cite{Korso2012}, are based on the above model.
} 
\textcolor{blue}{A more accurate scalar model is considered in \cite{Hu2018b}. Here, the authors assume that the received signal can be written as\cite[Eq. (2)]{Hu2018b} 
\begin{equation}
\label{Lund1}
 s({\bf r}) = \beta G(r) \sqrt{\dfrac{x_{C}}{r}}
 \end{equation}
 where the additional factor $\sqrt{x_{C}/r}$ accounts for the angle-of-arrival of the transmitted signal. The same model is adopted in \cite{AlegriaLISPositioning2019}.}
 
\blue{\subsection{The planar wave approximation}
In the region $\mathcal R_o$, we have that ${\bf r} ={\bf r}_C + {\bf d}$, where ${\bf r}_C=-x_C{\bf \hat x}+y_C{\bf \hat y}+z_C{\bf \hat z}$ is the vector from $C$ to $O$ and ${\bf d}=y{\bf \hat y}+z{\bf \hat z}$ is the vector from $O$ to $P$. Accordingly, \eqref{SferCart1} reduces to
\begin{equation}
\label{r_rc_d}
r=\|{\bf r}_C + {\bf d}\|=r_C\sqrt{1+2({\bf \hat r}_C \cdot {\bf \hat d})\dfrac{d}{r_C}+\dfrac{d^2}{r_C^2}}
\end{equation}
with $r_C=\|{\bf r}_C\|$ and $d=\| {\bf d} \|$. In the case $r_C \gg d$, we can replace $r$ with $r_C$ in the denominator of \eqref{GreenScalare} and obtain the following approximation:
\begin{equation}\label{Gr1}
G(r) \approx -\imagunit k Z_0 \dfrac{\mathrm{e}^{-\imagunit kr}}{4 \pi r_C}.
\end{equation}
As for the exponent $\mathrm{e}^{-\imagunit kr}$, approximations may be obtained by considering the Taylor series expansion $\sqrt{1+x} \approx 1+x/2-x^2/8$, valid for small values of $x$. Applying this approximation to \eqref{r_rc_d} yields
\begin{equation}
\label{ }
r \approx r_C+({\bf \hat r}_C \cdot {\bf \hat d})d+[1-({\bf \hat r}_C \cdot {\bf \hat d})^2]d^2/2r_C
\end{equation}
and \eqref{Gr1} thus reduces to 
\begin{equation}\label{Gr2}
G(r) \approx -\imagunit k Z_0 \dfrac{\mathrm{e}^{-\imagunit kr_C}}{4 \pi r_C} \mathrm{e}^{-\imagunit k[d \cos \psi+\sin^2 \psi (d^2/2r_C)]}
\end{equation}
where we have called $\cos \psi=({\bf \hat r}_C \cdot {\bf \hat d})$. This expansion is called the \textit{Fresnel approximation}~\cite{Benjamin2019}. In the case $r_C \gg L^2/\lambda$, we can retain only the first-order term in the exponent of~\eqref{Gr2} and thus obtain
\begin{equation}\label{Gr3}
G(r) \approx -\imagunit k Z_0 \dfrac{\mathrm{e}^{-\imagunit kr_C}}{4 \pi r_C} \mathrm{e}^{-\imagunit kd \cos \psi}
\end{equation}
which represents the well-known \textit{planar} approximation of the spherical wave in \eqref{GreenScalare2}.}

%\begin{remark}\label{remark2}
%A different approach for estimating the position of $C$ is to make use of a {\textrm scalar}, instead of a vector, field. For example, one could use only one of the three components of ${\bf e}(\vect{r})$. This may simplify the analysis but would result in lower performance, i.e., a larger CRB. An alternative approach is to consider a scalar field that is related to the component of the Poynting vector perpendicular in each point to the planar region. This component is proportional to $ \|{\bf e}(\vect{r})\|^2 \sin \theta \cos \phi$, and the associated scalar field is
%\begin{align}\notag
%E& \triangleq \mathrm{e}^{-\imagunit kr} \sqrt{ \|{\bf e}(\vect{r})\|^2 \sin \theta \cos \phi} \\&=kZ_0 \dfrac{\mathrm{e}^{-\imagunit kr} \sqrt{x_C}}{4 \pi r^{3/2}}  \sqrt{R^2_{\theta}(\theta,\phi)+ R^2_{\phi}(\theta,\phi)}.\label{PVScalarField}
%\end{align}
%In the case of an isotropic radiating source, $R^2_{\theta}(\theta,\phi)+ R^2_{\phi}(\theta,\phi)$ is independent of $\theta$ and $\phi$, and thus ~\eqref{PVScalarField} reduces to the scalar model considered in~\cite[Eq. (2)]{Hu2018b}. We stress that this scalar model represents a specific case, which is not valid in general. We will use~\eqref{PVScalarField} in the numerical analysis for comparisons.
%\end{remark}

% Analysis with a priori information
 %!TEX root = main.tex
 
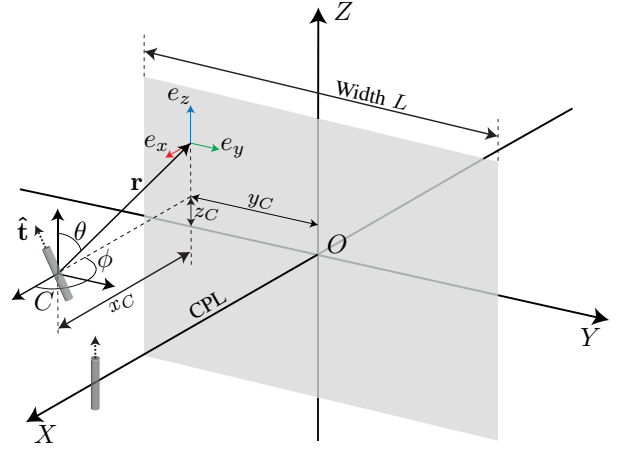
\begin{figure}[t!]
\centering
    \begin{overpic}[width=1\columnwidth]{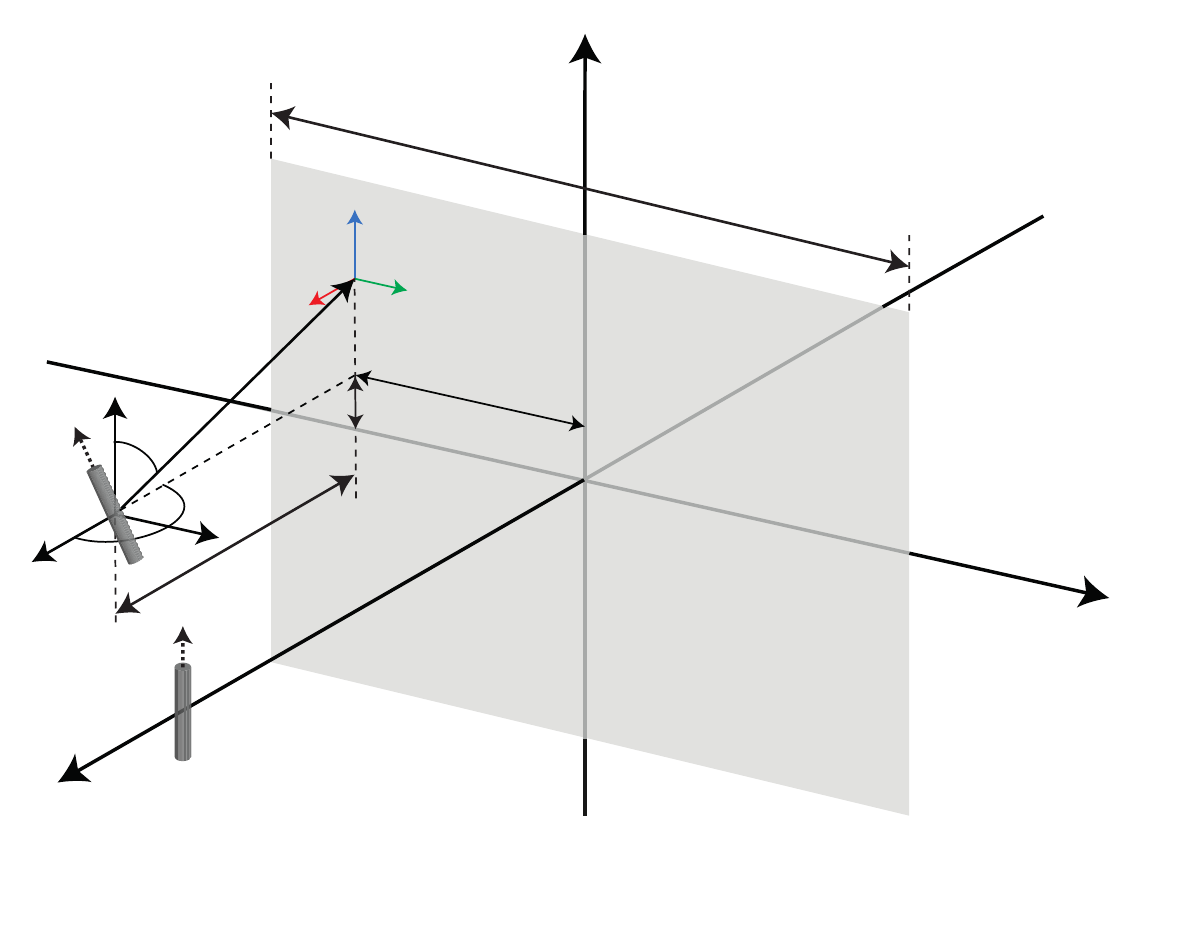}
    \put(50,38){$O$} %52 40
    %\put(5,16){$\widehat{\bf x}$}
    \put(88,24){$Y$}
    \put(51,73){$Z$}
    
    \put(6,29.5){$C$}
    \put(6,9.5){$X$}
%    \put(18,13){$Y'$}
%    \put(9,25){$Z'$}

   \put(3.5,40){$\mathbf{\hat{t}}$}

    % \put(37.5,51){p}
    \put(22.5,54){$\color{red}{e_x}$}
    \put(34,53.5){$\color{teal}{e_y}$}
    \put(26,61){$\color{blue}{e_z}$}
    
    \put(20.5,47.5){$\mathbf{r}$}
    \put(15.5,36){$\phi$}
    \put(12,40){$\theta$}
    
    \put(51,62){\rotatebox{-13}{\footnotesize{Width $L$}}}
    \put(29,28){\rotatebox{30}{\footnotesize\text{CPL}}}
    \put(17,29){\rotatebox{30}{\footnotesize\text{$ x_C$}}}
    \put(38,45.5){\rotatebox{-5}{\footnotesize\text{$ y_C$}}}
    \put(30,43.5){\rotatebox{0}{\footnotesize\text{$ z_C$}}}

%	%% t coordinates
%    \put(1,28){\rotatebox{0}{\footnotesize\text{$ X_t$}}}
%    \put(18,34){\rotatebox{0}{\footnotesize\text{$ Y_t$}}}
%    \put(5,44){\rotatebox{0}{\footnotesize\text{$ Z_t$}}}
%	

    \end{overpic}
   \caption{Illustration of an elementary dipole pointing in an arbitrary direction $\mathbf{\hat{t}} = t_x {\bf \hat x}+t_y {\bf \hat y}+t_z {\bf \hat z} $. The centroid of the dipole has cartesian coordinates $(x_C, y_C, z_C)$. The CPL case in which the dipole is located in the central perpendicular line (CPL) and is vertically oriented is also reported.}\vspace{-0.4cm}
%    \caption{In this figure, we have a square surface with side length equal to $L$. Two dipoles are represented, the first on the central perpendicular line (CPL) and the second at a generic point $\mathbf{S}$. The electric field of the latter is received at a generic point $\mathbf{p}$ belonging to the surface, and its three components are shown.}
    \label{fig:holo-surface}
\end{figure}

\section{CRB computation  with an elementary source dipole}
\label{sec:est3D}
To simplify the analysis, we make the following assumption about the source, as shown in Fig.~\ref{fig:holo-surface}.
\textcolor{orange}{
\begin{assumption}
\label{assumption-3}
The source is an elementary dipole of length $l_s$ pointing in an arbitrary direction $\mathbf{\hat{t}} = t_x {\bf \hat x}+t_y {\bf \hat y}+t_z {\bf \hat z}$.
\end{assumption}
In the case of an elementary dipole, the current density ${\bf j}({\bf s})$ has the following expression (e.g.,~\cite[Sec. 2.3.1]{ChewBook})
\begin{equation}
\label{J_dipole}
{\bf j}({\bf s})= I_{in} l_s \delta({\bf s}) \mathbf{\hat{t}}
\end{equation}
where $I_{in}$ is the uniform current level in the dipole. Plugging~\eqref{J_dipole} into \eqref{RadVect} yields
\begin{align}
\label{RadVecDipole_1} {\bf R}(\theta,\phi)&=I_{in} l_s \mathbf{\hat{t}}
\end{align}
from which, using \eqref{E_P.0}, we have
\begin{align}
\label{E_P_PTAI_0}
{\bf \mathbf{e} }({\bf r}) = G(r)  I_{in} l_s [\left(\mathbf{\hat{r}} \mathbf{\times} \mathbf{\hat{t}}\right)\mathbf{\times} \mathbf{\hat{r}}].
\end{align}
By using the identity $\left(\mathbf{\hat{r}} \mathbf{\times} \mathbf{\hat{t}}\right)\mathbf{\times} \mathbf{\hat{r}}=\mathbf{\hat{t}}-(\mathbf{\hat{r}}\cdot \mathbf{\hat{t}}) \mathbf{\hat{r}}$, \eqref{E_P_PTAI_0} can be written in the equivalent form
\begin{align}
\label{E_P_PTAI}
{\bf \mathbf{e} }({\bf r}) = G(r)  I_{in} l_s [\mathbf{\hat{t}}-(\mathbf{\hat{r}}\cdot \mathbf{\hat{t}}) \mathbf{\hat{r}}].
\end{align}}
From~\eqref{GreenScalare2} and~\eqref{E_P_PTAI}, it thus follows that
\begin{align}
\label{CartesianComponents_PTAI_x}
e_x&= -\imagunit \chi \dfrac{\mathrm{e}^{-\imagunit kr}}{r} [t_x-(r_x t_x+r_y t_y+r_z t_z)r_x]\\\label{CartesianComponents_PTAI_y}
e_y&= -\imagunit \chi \dfrac{\mathrm{e}^{-\imagunit kr}}{r} [t_y-(r_x t_x+r_y t_y+r_z t_z)r_y]\\\label{CartesianComponents_PTAI_z}
e_z&= -\imagunit \chi \dfrac{\mathrm{e}^{-\imagunit kr}}{r} [t_z-(r_x t_x+r_y t_y+r_z t_z)r_z].
\end{align}
where $(r_x,r_y,r_z)$ are the cartesian components of $\mathbf{\hat{r}}$, and
\begin{align}\label{eq:chi}
\chi= \dfrac{Z_0 I_{in}}{2} \dfrac{l_s}{\lambda}
\end{align}
is measured in volts, $\mathrm V$. For the sake of simplicity, we have dropped the dependence of $e_x$, $e_y$ and $e_z$ on $\bf r$. The above expressions depend on $(t_x,t_y,t_z)$ and the cartesian coordinates $(x_C, y_C, z_C)$ of the centroid $C$, which are related to $(r_x,r_y,r_z)$ through the following equations:
\begin{align}
\label{rcomp}
r_x=-\dfrac{x_C}{r}, \quad r_y=\dfrac{y-y_C}{r}, \quad r_z=\dfrac{z-z_C}{r}
\end{align}
where $r$ is given by \eqref{SferCart1}. 
%\textcolor{blue} {
%\begin{remark}
%%The components $R_{\theta}(\theta,\phi)$ and $R_{\phi}(\theta,\phi)$ of the radiation vector can be expressed in terms of $t_x$, $t_y$ and $t_z$, by using the following relationships between spherical and cartesian components:
%%\begin{align}
%%\label{tetaxyz}
%%&t_{\theta}(\theta,\phi)=\cos\theta \cos\phi \, t_x+\cos\theta \sin\phi \, t_y -\sin\theta \, t_z
%%\\
%%\label{fixyz}
%%&t_{\phi}(\theta,\phi)=- \sin\phi \, t_x+\cos\phi \, t_y
%%\end{align}
%Substituting \eqref{tetaxyz}--\eqref{fixyz} into \eqref{RadVecDipole2}--\eqref{RadVecDipole3} yields
%\begin{align}
%\label{RadVecDipole2_bis} 
%R_{\theta}(\theta,\phi)&=I_{in} l_s (\cos\theta \cos\phi \, t_x+\cos\theta \sin\phi \, t_y -\sin\theta \, t_z)\\
%\label{RadVecDipole3_bis} 
%R_{\phi}(\theta,\phi)&=I_{in} l_s (- \sin\phi \, t_x+\cos\phi \, t_y)
%\end{align}
%\end{remark}
%TO BE CONTINUED
%}

We aim at evaluating the CRBs for the estimation of $(x_C,y_C,z_C)$ on the basis of~\eqref{CartesianComponents_PTAI_x}--\eqref{CartesianComponents_PTAI_z}. 
We first consider the case in which the parameters $(t_x,t_y,t_z)$ are \textit{unknown} and thus must be considered as \textit{nuisance} parameters. This corresponds to the case of having \textit{partial} information about the source current distribution. This first case is not only more general but also instrumental to obtain the CRBs with \emph{full} information, i.e., the parameters $(t_x,t_y,t_z)$ are a priori \textit{known}.

%{\color{blue} DIREI DI RIFARE LA FIGURA 2 CON IL DIPOLO ORIENTATO IN MANIERA ARBITRARIA, INDICANDO ANCHE IL VERSORE $\mathbf{\hat{t}}$. LA FIGURA CON IL DIPOLO ORIENTATO VERTICALMENTE LA POSSIAMO USARE QUANDO ANALIZZIAMO IL CASO PARTICOLARE.}
%The system setup is shown in Fig.~\ref{fig:holo-surface}. 

\subsection{Unknown orientation of the dipole}
%We start our analysis by considering the case in which the arbitrary direction $\mathbf{\hat{t}}$ is \textit{unknown}. This corresponds to the case we only have \textit{partial} information about the source current distribution, since we know that the transmit antenna is an elementary dipole but its orientation is unknown. 
When $(t_x,t_y,t_z)$ are unknown, we cannot use \eqref{eq:CRBs} -- \eqref{fmn} but must compute the FIM for \textit{all} the unknown parameters, which are collected into the $6$-dimensional vector $\mathbf{p}=(t_x,t_y,t_z,x_C,y_C,z_C)$. Therefore, the FIM is a $6\times6$ hermitian matrix with entries given by {
\begin{align}
\notag
&\big[\mathbf{F}\big]_{mn}=\\&\dfrac{2}{\sigma^2}\re\left\{\iint_{-\frac{L}{2}}^{\frac{L}{2}} \left[\derpar{e_x}{p_m} \derpar{e_x^\ast}{p_n}+\derpar{e_y}{p_m} \derpar{e_y^\ast}{p_n}+\derpar{e_z}{p_m} \derpar{e_z^\ast}{p_n}\right]dy dz\right\}\label{eq:elemF_6D}
\end{align}}
where $p_m$ denotes the $m$th element of $\mathbf{p}$ and $\re\{\cdot\}$ is the real part of the enclosed quantity. The derivatives involved in $[\mathbf{F}]_{mn}$ are computed in Appendix A. The matrix $\mathbf{F}$ can be partitioned as~\cite{Kay1993a}
\begin{align}
\label{matF}
\mathbf{F}=\left[\begin{array}{c|c} \mathbf{F}_{tt} & \mathbf{F}_{tc} \\ \hline \mathbf{F}_{ct} & \mathbf{F}_{cc}\end{array}\right]
\end{align}
where the $3 \times 3$ blocks $\mathbf{F}_{tt}$ and $\mathbf{F}_{cc}$ contain the partial derivatives with respect to $(t_x,t_y,t_z)$ and $(x_C, y_C, z_C)$, respectively, while $\mathbf{F}_{tc}$ and $\mathbf{F}_{ct}$ contain the mixed derivatives. Since $\mathbf{F}$ is symmetric, we have $\mathbf{F}_{tt}=\mathbf{F}^T_{tt}$,  $\mathbf{F}_{cc}=\mathbf{F}_{cc}^T$ and $\mathbf{F}_{tc}=\mathbf{F}_{ct}^T$. Based on \eqref{matF} and well known formulas \blue{on the inverse of partitioned matrices~\cite[Sec. A1.1.3]{Kay1993a}}, we can immediately show that the CRBs for the estimation of $x_C$, $y_C$ and $z_C$, are given by the diagonal elements of the matrix $\left(\mathbf{F}_{cc}-\mathbf{F}_{tc}^T\mathbf{F}_{tt}^{-1}\mathbf{F}_{tc}\right)^{-1}$, i.e.,
%\begin{equation}
%\label{matC}
%{\bf C}=\left(\mathbf{F}_{cc}-\mathbf{F}_{tc}^H\mathbf{F}_{tt}^{-1}\mathbf{F}_{tc}\right)^{-1}
%\end{equation}
%In particular,
\begin{align}
\label{CRB_x_u}
{\rm{CRB}}_{\rm{u}}(x_C)=\left[\left(\mathbf{F}_{cc}-\mathbf{F}_{tc}^T\mathbf{F}_{tt}^{-1}\mathbf{F}_{tc}\right)^{-1}\right]_{11}  \\ 
\label{CRB_y_u} {\rm{CRB}}_{\rm{u}}(y_C)=\left[\left(\mathbf{F}_{cc}-\mathbf{F}_{tc}^T\mathbf{F}_{tt}^{-1}\mathbf{F}_{tc}\right)^{-1}\right]_{22}  \\
\label{CRB_z_u}  {\rm{CRB}}_{\rm{u}}(z_C)=\left[\left(\mathbf{F}_{cc}-\mathbf{F}_{tc}^T\mathbf{F}_{tt}^{-1}\mathbf{F}_{tc}\right)^{-1}\right]_{33} 
\end{align}
where the subscript $_{\rm{u}}$ is used to stress that the above results refer to the case of \emph{unknown} dipole orientation.

\begin{remark}
{\color{orange} Notice that the electromagnetic model in~\eqref{CartesianComponents_PTAI_x}-~\eqref{CartesianComponents_PTAI_z} can also be used to compute the CRBs for the cartesian components $(t_x,t_y,t_z)$ of $\mathbf{\hat{t}}$. Evaluating these bounds is out of the scope of this work whose focus is the estimation of the source position. However, we point out that estimating the orientation may be useful in practice, e.g., for the deployment and orientation of receiving surfaces.}
\end{remark}

\subsection{Known orientation of dipole}
When the arbitrary parameters $(t_x,t_y,t_z)$ of the dipole are perfectly known, we have a complete description of the source current distribution. This means that the functions $R_{\theta}(\theta,\phi)$ and $R_{\phi}(\theta,\phi)$ are known, and we can compute the CRB following the general\footnote{Here, by \textit{general} we mean that it is valid irrespective of the particular current distribution, as long as the latter is known.} procedure outlined in Section \ref{CRB_full}. The Fisher's information matrix obtained in this way coincides with the matrix $\mathbf{F}_{cc}$ computed previously. It thus follows that the CRBs for the estimation of $x_C$, $y_C$ and $z_C$, can be found as the diagonal elements of the matrix $\mathbf{F}^{-1}_{cc}$, i.e.,
\begin{align}
\label{matFcc_CRBx}
{\rm{CRB}}(x_C)&=\left[{\bf F}^{-1}_{cc}\right]_{11}  \\ 
\label{matFcc_CRBy} {\rm{CRB}}(y_C)&=\left[{\bf F}^{-1}_{cc}\right]_{22}  \\ 
\label{matFcc_CRBz} {\rm{CRB}}(z_C)&=\left[{\bf F}^{-1}_{cc}\right]_{33}. 
\end{align}
\vspace{-1cm}
\textcolor{blue}{\subsection{Relationship between ${\rm{CRBs}}$ with and without knowledge of the dipole orientation}
By applying the matrix inversion lemma,\footnote{$\left(\mathbf{A}+\mathbf{U}\mathbf{C}\mathbf{V}\right)^{-1} = \mathbf{A}^{-1} - \mathbf{A}^{-1}\mathbf{U}(\mathbf{C}^{-1}+\mathbf{V}\mathbf{A}^{-1}\mathbf{U})^{-1}\mathbf{V}\mathbf{A}^{-1}$} we obtain
\begin{equation}
\label{MIL}
\begin{split}
&\left[\left(\mathbf{F}_{cc}-\mathbf{F}_{tc}^T\mathbf{F}_{tt}^{-1}\mathbf{F}_{tc}\right)^{-1}\right]_{ii}\\=&\left[\mathbf{F}_{cc}^{-1}\right]_{ii}+\left[\mathbf{F}_{cc}^{-1}\mathbf{F}_{tc}^T\left(\mathbf{F}_{tt}-\mathbf{F}_{tc}\mathbf{F}_{cc}^{-1}\mathbf{F}_{tc}^T\right)^{-1}\mathbf{F}_{tc}\mathbf{F}_{cc}^{-1}\right]_{ii}
\end{split}
\end{equation}
for $i=1,2,3$. Notice that $\left(\mathbf{F}_{tt}-\mathbf{F}_{tc}\mathbf{F}_{cc}^{-1}\mathbf{F}_{tc}^T\right)^{-1}$ is positive semi-definite, since it is the $3 \times 3$ block in the upper left corner of the semi-positive definite matrix ${\bf F}^{-1}$. As a result, the second matrix $\mathbf{F}_{cc}^{-1}\mathbf{F}_{tc}^T\left(\mathbf{F}_{tt}-\mathbf{F}_{tc}\mathbf{F}_{cc}^{-1}\mathbf{F}_{tc}^T\right)^{-1}\mathbf{F}_{tc}\mathbf{F}_{cc}^{-1}$ in~\eqref{MIL} is positive semi-definite as well. Hence, from \eqref{MIL} it follows that
\begin{equation}
\label{MILbis}
\begin{split}
\left[\mathbf{F}_{cc}^{-1}\right]_{ii} \le \left[\left(\mathbf{F}_{cc}-\mathbf{F}_{tc}^T\mathbf{F}_{tt}^{-1}\mathbf{F}_{tc}\right)^{-1}\right]_{ii}. 
\end{split}
\end{equation}
By exploiting this result, we have that
\begin{align}
\label{Inx}
{\rm{CRB}}(x_C) \le {\rm{CRB}}_{\rm{u}}(x_C) \\ 
\label{Iny}
{\rm{CRB}}(y_C) \le {\rm{CRB}}_{\rm{u}}(y_C) \\
\label{Inz} 
{\rm{CRB}}(z_C) \le {\rm{CRB}}_{\rm{u}}(z_C) 
\end{align}
as it should be since the Cram\'er-Rao bound worsens in the presence of unknown nuisance parameters~\cite[Example 3.7]{Kay1993a}. However, numerical results will show that the loss in estimation accuracy between the two cases in which the parameters $(t_x,t_y,t_z)$ are unknown and known may be negligible under certain conditions. More details will be given in Section~\ref{Sec:Numerical_analysis}.}

%The inequalities \eqref{Inx}--\eqref{Inz} straightforwardly derive from the matrix inversion lemma (MIL). Indeed, by applying MIL to \eqref{CRB_x_u}--\eqref{CRB_z_u}, we have
%\begin{equation}
%\label{MIL}
%\begin{split}
%&\left[\left(\mathbf{F}_{cc}-\mathbf{F}_{tc}^H\mathbf{F}_{tt}^{-1}\mathbf{F}_{tc}\right)^{-1}\right]_{ii}\\=&\left[\mathbf{F}_{cc}^{-1}\right]_{ii}+\left[\mathbf{F}_{cc}^{-1}\mathbf{F}_{tc}^H\left(\mathbf{F}_{tt}-\mathbf{F}_{tc}\mathbf{F}_{cc}^{-1}\mathbf{F}_{tc}^H\right)^{-1}\mathbf{F}_{tc}\mathbf{F}_{cc}^{-1}\right]_{ii}
%\end{split}
%\end{equation}
%for $i=1,2,3$. Now, observe that $\left(\mathbf{F}_{tt}-\mathbf{F}_{tc}\mathbf{F}_{cc}^{-1}\mathbf{F}_{tc}^H\right)^{-1}$ is positive semi-definite, since it is the $3 \times 3$ block in the upper left corner of the semi-positive definite matrix ${\bf F}^{-1}$. As a result, matrix $\mathbf{F}_{cc}^{-1}\mathbf{F}_{tc}^H\left(\mathbf{F}_{tt}-\mathbf{F}_{tc}\mathbf{F}_{cc}^{-1}\mathbf{F}_{tc}^H\right)^{-1}\mathbf{F}_{tc}\mathbf{F}_{cc}^{-1}$ is positive semi-definite as well, and hence from \eqref{MIL} we have
%\begin{equation}
%\label{MIL}
%\begin{split}
%\left[\mathbf{F}_{cc}^{-1}\right]_{ii} \le \left[\left(\mathbf{F}_{cc}-\mathbf{F}_{tc}^H\mathbf{F}_{tt}^{-1}\mathbf{F}_{tc}\right)^{-1}\right]_{ii} 
%\end{split}
%\end{equation}
%which are equivalent to \eqref{Inx}--\eqref{Inz}.

%\section{A case study: analysis in the CPL case and a vertically oriented dipole}
\section{A case study}\label{Sec:case-study}
The elements of the matrices needed for the computation of $\mathbf{F}$ in~\eqref{matF} can be obtained by numerically evaluating the integrals~\eqref{eq:elemF_6D} for any arbitrary position of the dipole. Although possible, this makes it hard to gain insights into the CRBs. Next, we show that closed-form expressions can be obtained when the dipole is located in the central perpendicular line (CPL) and is vertically oriented; see Fig.~\ref{fig:holo-surface}. The following assumption is thus made.
\begin{assumption}
\label{ass:cpl_vertical}
The center $C$ of the dipole is on the line perpendicular to  $\mathcal R_o$ passing through the point $O$ (known as CPL case), as shown in Fig.~\ref{fig:holo-surface}, and the dipole is vertically oriented. 
\end{assumption}
Under the CPL assumption, we have that $y_C=z_C=0$, whereas the vertical orientation of the dipole implies that $t_x=t_y=0$ and $t_z=1$.
For convenience, we call 
\begin{align}\mathrm{SNR} \triangleq \frac{2 |\chi|^2}{\sigma^2}
\end{align}
 the signal-to-noise ratio (SNR) with $\chi$ given by~\eqref{eq:chi} and define
\begin{align}\label{rho}
 \rho \triangleq \frac{L}{x_C}.
 \end{align}
%{\color{blue}CONVIENE INSERIRE UNA NUOVA FIGURA CON IL DIPOLO ORIENTATO VERTICALMENTE}x
The following result is obtained.

\begin{proposition}
\label{CPL_VD}
Under Assumption~\ref{ass:cpl_vertical}, we have that:
\begin{enumerate}
  \item The matrix $\mathbf{F}_{cc}$ becomes diagonal with entries
\begin{align}
 \label{eq:Fcc11_A4}
[\mathbf{F}_{cc}]_{11}&= \mathrm{SNR} \cdot \left[k^2 \mathscr{I}_1(\rho) +  x_C^{-2} \mathscr{I}_2(\rho)\right]\\ 
\label{eq:Fcc22_A4}
[\mathbf{F}_{cc}]_{22}&= \mathrm{SNR} \cdot \left[k^2 \mathscr{I}_3(\rho) +  x_C^{-2} \mathscr{I}_4(\rho)\right]\\ 
\label{eq:Fcc33_A4}
[\mathbf{F}_{cc}]_{33}&= \mathrm{SNR} \cdot \left[k^2 \mathscr{I}_5(\rho) +  x_C^{-2} \mathscr{I}_6(\rho)\right]
\end{align}
where 
$\mathscr{I}_i(\rho)$ for $i = 1, \ldots, 6$ are given in Appendix B.

\item The matrix $\mathbf{F}_{tt}$ becomes diagonal with entries
\begin{align}
\label{Ftt_11}
[\mathbf{F}_{tt}]_{11}&=\mathrm{SNR} \cdot \mathscr{I}_7(\rho)\\
\label{Ftt_22}
[\mathbf{F}_{tt}]_{22}&=[\mathbf{F}_{tt}]_{33}=\mathrm{SNR} \cdot  \mathscr{I}_8(\rho)
\end{align}
where $\mathscr{I}_i(\rho)$ for $i = 7, 8$ are given in \eqref{I7}--\eqref{I8} in Appendix B.

\item The elements of $\mathbf{F}_{tc}$ are all zero except for $[\mathbf{F}_{tc}]_{13}$ and $[\mathbf{F}_{tc}]_{31}$ which can be computed as 
\begin{align}
\label{Ftc13}
[\mathbf{F}_{tc}]_{13}&=x^{-1}_{C} \mathrm{SNR} \cdot \mathscr{I}_9(\rho)\\
\label{Ftc31}
[\mathbf{F}_{tc}]_{31}
&= -x^{-1}_{C} \mathrm{SNR} \cdot \mathscr{I}_{10}(\rho)
\end{align}
where $\mathscr{I}_{9}(\rho)$ and $\mathscr{I}_{10}(\rho)$ are given, respectively, by \eqref{I9} and \eqref{I10} in Appendix B.

\end{enumerate}

\end{proposition}

\begin{proof}
The proof is given in Appendix B.
\end{proof}

Based on the results in Proposition \ref{CPL_VD}, the following corollaries are obtained for both cases with unknown and known dipole orientation. Notice that the dependence on $\rho$ is omitted to simplify the notation.
\begin{corollary}[Unknown dipole orientation]\label{Corollary1}Under Assumption~\ref{ass:cpl_vertical}, the CRBs for the estimation of $x_C$, $y_C$ and $z_C$ when the dipole orientation is unknown are given by
\begin{align}
\label{eq:crb_ori_x}
\mathrm{CRB}_{\rm{u}}(x_C)&=\dfrac{\mathrm{SNR}^{-1}}{k^2 \mathscr{I}_1 +  x_C^{-2} (\mathscr{I}_2-\mathscr{I}_8^{-1}\mathscr{I}_{10}^2)}\\
\label{eq:crb_ori_y}
\mathrm{CRB}_{\rm{u}}(y_C)&=\dfrac{\mathrm{SNR}^{-1}}{k^2\mathscr{I}_3 + x_C^{-2}\mathscr{I}_4}\\
\label{eq:crb_ori_z}
\mathrm{CRB}_{\rm{u}}(z_C)&=\dfrac{\mathrm{SNR}^{-1}}{k^2\mathscr{I}_5 +  x_C^{-2}(\mathscr{I}_6-\mathscr{I}_7^{-1}\mathscr{I}_{9}^2)}.
\end{align}
\end{corollary}
\begin{proof}
The proof follows from \eqref{CRB_x_u} -- \eqref{CRB_z_u} by using the results of Proposition~\ref{CPL_VD} from which we have 
%$\mathbf{F}_{tc}^H\mathbf{F}_{tt}^{-1}\mathbf{F}_{tc}$ is a $3 \times 3$ diagonal matrix given by
$\mathbf{F}_{tc}^T\mathbf{F}_{tt}^{-1}\mathbf{F}_{tc} = \diag\left([\mathbf{F}_{tc}]_{31}^2/{[\mathbf{F}_{tt}]_{33}},0, [\mathbf{F}_{tc}]_{13}^2/{[\mathbf{F}_{tt}]_{11}}\right)$.
\end{proof}

\begin{corollary}[Known dipole orientation]
Under Assumption~\ref{ass:cpl_vertical}, the CRBs for the estimation of $x_C$, $y_C$ and $z_C$ when the dipole orientation is known are:
\begin{align}
\label{eq:crb_x}
\mathrm{CRB}(x_C)&= \dfrac{\mathrm{SNR}^{-1}}{k^2 \mathscr{I}_1 +  x_C^{-2} \mathscr{I}_2}\\ \label{eq:crb_y}
\mathrm{CRB}(y_C)&=\dfrac{\mathrm{SNR}^{-1}}{k^2\mathscr{I}_3 + x_C^{-2}\mathscr{I}_4}\\ \label{eq:crb_z}
\mathrm{CRB}(z_C)&= \dfrac{\mathrm{SNR}^{-1}}{k^2\mathscr{I}_5 +  x_C^{-2}\mathscr{I}_6}.	
\end{align}
\end{corollary}

\begin{proof}
It easily follows from Proposition~\ref{CPL_VD} by using \eqref{matFcc_CRBx}--\eqref{matFcc_CRBz}.
\end{proof}

%Equations \eqref{eq:crb_ori_x}--\eqref{eq:crb_ori_z} and \eqref{eq:crb_x}--\eqref{eq:crb_z} 
The above corollaries clearly show the effects of the wavelength $\lambda=2 \pi /k$ and $x_C$ for fixed values of $\rho$ or, equivalently, of the functions $\{\mathscr{I}_i\}$. Particularly, we see that the estimation accuracy increases as $\lambda$ or $x_C$ decrease for fixed values of $\mathrm{SNR}$. Similar results were already observed in \cite{Hu2018b}. Also, from \eqref{eq:crb_ori_y} and \eqref{eq:crb_y} it follows that, under Assumption~\ref{ass:cpl_vertical}, $\mathrm{CRB}_{\rm{u}}(y_C)=\mathrm{CRB}(y_C)$. 
%This means that the estimation accuracy achieved with or without knowledge of the source orientation is asymptotically the same.

 \subsection{Analysis for $x_C\gg\lambda$}

Assumption~\ref{ass:cpl_vertical} leads to a considerable simplification since the matrices $\mathbf{F}_{cc}$ and $\mathbf{F}_{tc}^T\mathbf{F}_{tt}^{-1}\mathbf{F}_{tc}$ become diagonal. Further simplifications can be obtained when $x_C\gg\lambda$. This condition is always satisfied in systems operating at frequencies in the range of GHz or above. %(which in general holds, as observed in \cite{Hu2018b}, since $\lambda$ is the wavelength) by the following proposition:
%\begin{proposition} 
%\label{prop:xcgglambda}
%Under Assumption~\ref{ass:cpl_vertical} and $x_C\gg\lambda$, then CRBs for the case in which the dipole orientation is unknown reduce to
%\begin{align}
%\label{eq:crb_ori_x_app1}
%\mathrm{CRB}_{\rm{u}}(x_C) &\approx \dfrac{\mathrm{SNR}^{-1}}{\mathscr{I}_1-\mathscr{I}_8^{-1}\left|\mathscr{I}_{10}\right|^2} \cdot \dfrac{\lambda^2}{4 \pi^2}\\
%\label{eq:crb_ori_y_app1}
% \mathrm{CRB}_{\rm{u}}(y_C)&\approx \dfrac{\mathrm{SNR}^{-1}}{\mathscr{I}_3} \cdot \dfrac{\lambda^2}{4 \pi^2}
%\end{align}
%while those for the case in which the dipole orientation is known become
%\begin{align}
%\label{eq:crb_x_app1}
%\mathrm{CRB}(x_C) &\approx \dfrac{\mathrm{SNR}^{-1}}{\mathscr{I}_1} \cdot \dfrac{\lambda^2}{4 \pi^2}\\
%\label{eq:crb_y_app1}
%\mathrm{CRB}(y_C) &= \mathrm{CRB}_{\rm{u}}(y_C).
%\end{align} 
%\end{proposition}
\begin{proposition} 
\label{prop:xcgglambda}
Under Assumption~\ref{ass:cpl_vertical} and $x_C\gg\lambda$, CRBs reduce to
\begin{align}
\label{eq:crb_ori_x_app1}
\mathrm{CRB}_{\rm{u}}(x_C) \approx \mathrm{CRB}(x_C) &\approx \dfrac{\mathrm{SNR}^{-1}}{\mathscr{I}_1} \cdot \dfrac{\lambda^2}{4 \pi^2}\\
\label{eq:crb_ori_y_app1}
 \mathrm{CRB}_{\rm{u}}(y_C)= \mathrm{CRB}(y_C)&\approx \dfrac{\mathrm{SNR}^{-1}}{\mathscr{I}_3} \cdot \dfrac{\lambda^2}{4 \pi^2}
\end{align} 
\end{proposition}
\begin{proof}
See Appendix C.
\end{proof}

Proposition \ref{prop:xcgglambda} shows that, when $x_C \gg \lambda$, the accuracy in the estimation of $x_C$ and $y_C$ solely depends on the values of $\lambda$ and $\rho$. In particular, keeping $\rho$ fixed, the CRBs for $x_C$ and $y_C$ scale with the square of $\lambda$. On the other hand, for a given value of $\lambda$, if $x_C$ increases by a factor $\alpha$, we must scale $L$ by the same factor in order to keep $\rho$, and hence the estimation accuracy, unchanged. This means that the area of the observation region must be increased by a factor $\alpha^2$. The same conclusions do not hold for the estimation of $z_C$. Indeed, under Assumption~\ref{ass:cpl_vertical} and $x_C\gg\lambda$ it can be shown that the terms $x_C^{-2}(\mathscr{I}_6-\mathscr{I}_7^{-1}\mathscr{I}_{9}^2)$ and $x_C^{-2}\mathscr{I}_6$, appearing in the denominator of the expressions for $\mathrm{CRB}_u(z_C)$ and $\mathrm{CRB}(z_C)$ in Corollaries 1 and 2, are not negligible with respect to $k^2\mathscr{I}_5$. Hence, the CRBs cannot be simplified and continue to depend on both $x_C$ and $\lambda$.

\subsection{Asymptotic analysis for $\rho \to \infty$}

Starting from the results given above, \blue{in order to understand the ultimate performance and to obtain insightful closed form expressions} it is useful to analyze the behaviour of the CRBs in the asymptotic regime $\rho = L/x_C\to \infty$. The main results are summarized in the following proposition.

\begin{proposition}
\label{rho_inf}
Under Assumption~\ref{ass:cpl_vertical} and $x_C \gg \lambda$, in the asymptotic regime $\rho \to \infty$ we have
\begin{align}
\label{eq:CRBXtoINF}
\!\!\!\!\lim\limits_{\rho \to \infty} \mathrm{CRB}_{\rm{u}}(x_C)=\lim\limits_{\rho \to \infty} \mathrm{CRB}(x_C)&=\dfrac{\mathrm{SNR}^{-1}} {3 \pi^3} \lambda^2\\
\label{eq:CRBYtoINF}
\!\!\!\!\lim\limits_{\rho \to \infty} \mathrm{CRB}_{\rm{u}}(y_C){\ln \rho}=\lim\limits_{\rho \to \infty}\mathrm{CRB}(y_C){\ln \rho}& = \dfrac{\mathrm{SNR}^{-1}} {3 \pi^3} \lambda^2\\ \label{eq:CRBZtoINF}
\!\!\!\!\lim\limits_{\rho \to \infty} \mathrm{CRB}_{\rm{u}}(z_C){\ln \rho}=\lim\limits_{\rho \to \infty} \mathrm{CRB}(z_C){\ln \rho}&  = \dfrac{\mathrm{SNR}^{-1}} { \pi^3} \lambda^2.\!\!
%\mathrm{CRB}_{\rm{u}}(y_C)=\mathrm{CRB}(y_C)& \sim \dfrac{\mathrm{SNR}^{-1}} {3 \pi^3} \dfrac{\lambda^2}{\ln \rho}\\ \label{eq:CRBZtoINF}
%\lim\limits_{\rho \to \infty} \mathrm{CRB}_{\rm{u}}(z_C)=\lim\limits_{\rho \to \infty} \mathrm{CRB}(z_C)&  \sim \dfrac{\mathrm{SNR}^{-1}} { \pi^3} \dfrac{\lambda^2}{\ln \rho}.
\end{align}
\end{proposition}

\begin{proof}
See Appendix D.
\end{proof}

Proposition~\ref{rho_inf} shows that, for sufficiently large values of $\rho = L/x_C$, the estimation accuracy is the same in both cases of unknown or known dipole orientation. This means that, even though $\mathbf{\hat{t}}$ is unknown, we can achieve the same accuracy in the estimation of the source position as in the case of known dipole orientation. Clearly, this requires in general the \textit{joint} estimation of $\mathbf{\hat{t}}$ and $(x_C,y_C,z_C)$. {\color{orange}It is interesting to note that the CRBs for the estimation of $y_C$ and $z_C$ goes to zero as $\rho$ increases unboundedly. This is in contrast to the results in~\cite[Eq. (26)]{Hu2018b} where it is shown that the asymptotic CRBs are identical for all the three dimensions and depend solely on the wavelength $\lambda$. This difference is a direct consequence of the different radiation and signal models used for the computation of CRBs. Indeed, in \cite{Hu2018b} the bounds are derived on the basis of the scalar field \eqref{Lund1}.}
%This difference is a direct consequence of the different radiation and signal models used for the computation of CRBs. Indeed, in \cite{Hu2018b} the source is assumed to radiate isotropically and the scalar field \eqref{Lund1} is used for deriving the bounds.

%{\color{red}
%% Analysis with no a priori information
%\input{SectionIV}
%}

% Analysis with no a priori information

 %!TEX root = main.tex

\begin{figure}[t!]
        \centering
	\begin{subfigure}[t]{1\columnwidth} \centering 
	\begin{overpic}[width=\columnwidth]{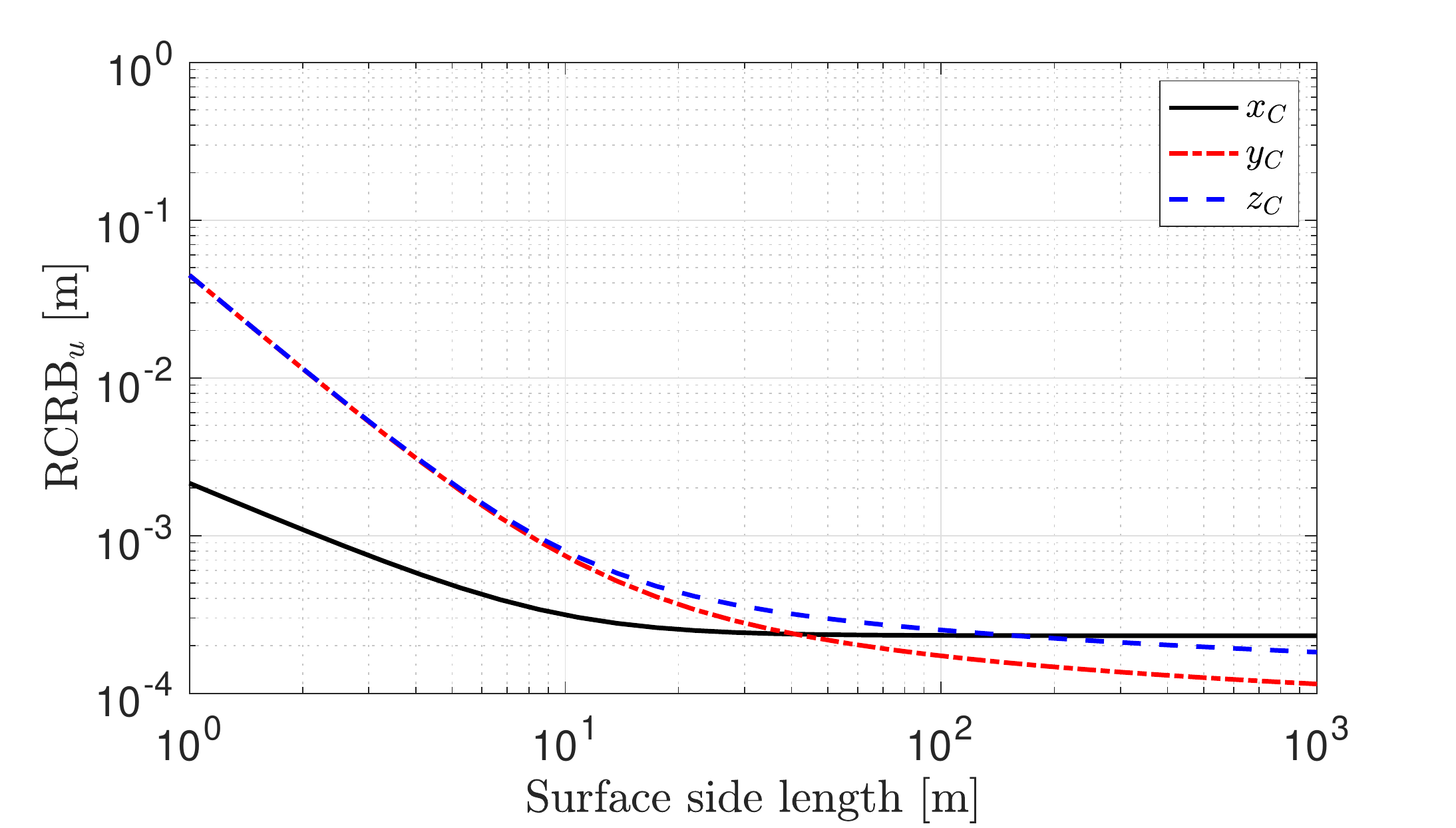}
	\put(17,11){\scriptsize{Asymptotic limit, $L \to \infty$}}
	\end{overpic} 	

	\caption{Unknown dipole orientation}\vspace{0.2cm}
	\label{fig:unk001}
\end{subfigure}
\begin{subfigure}[t]{1\columnwidth} \centering 
	\begin{overpic}[width=\columnwidth]{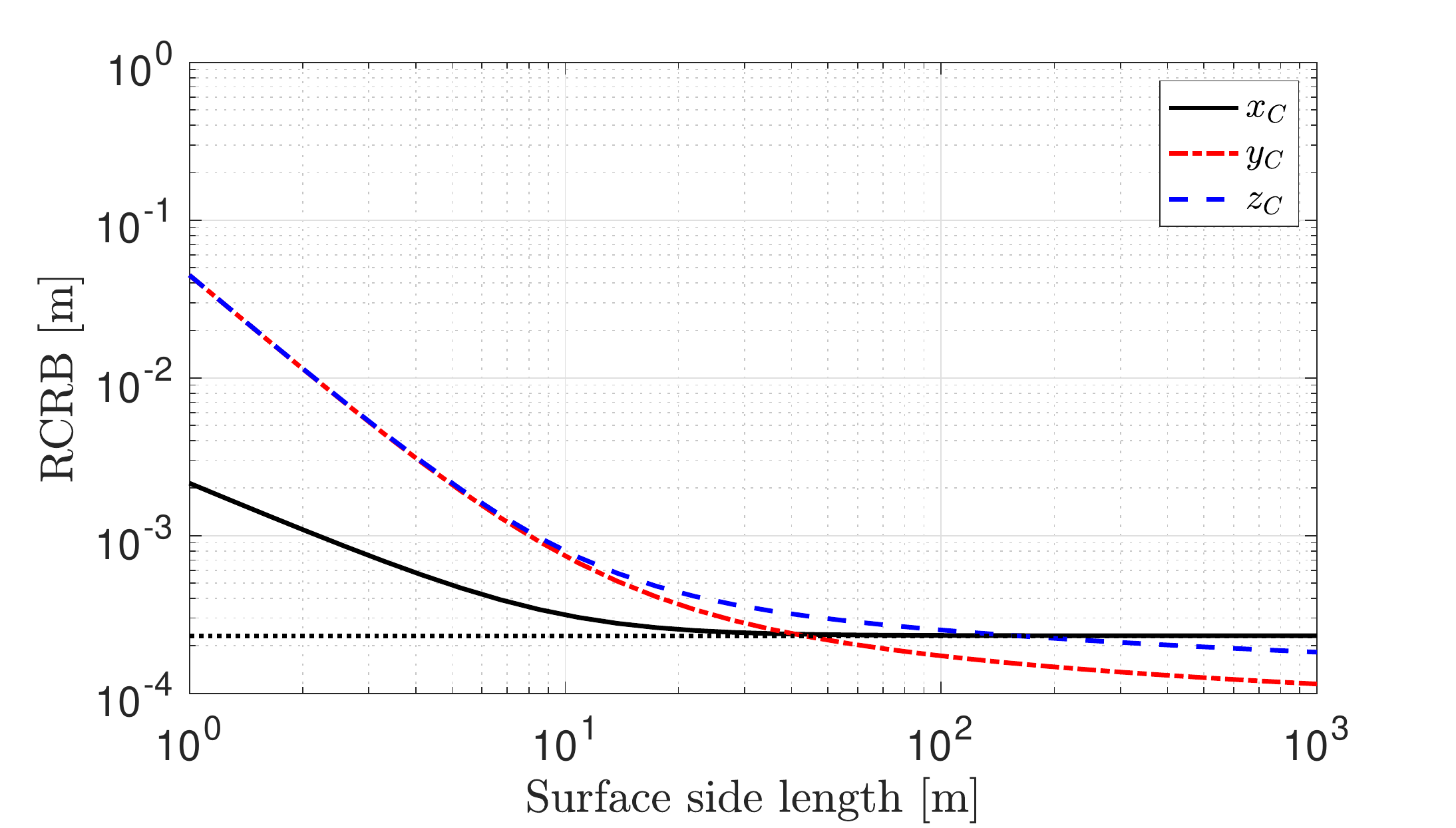}
	\put(17,10.7){\scriptsize{Asymptotic limit, $L \to \infty$}}
	\end{overpic} 	
	\caption{Known dipole orientation}\vspace{0.2cm}
	\label{fig:k001}
\end{subfigure}

\caption{RCRBs as a function of the surface side length for a vertically oriented dipole in CPL at a distance $x_C=6$ m.}
\label{fig:001_vs_size}
\end{figure}

\section{Numerical Analysis of Cram{\'e}r-Rao Bounds}\label{Sec:Numerical_analysis}

The effect of the various system parameters on the estimation accuracy are now analyzed by numerically evaluating the CRBs according to the general expressions given in \eqref{CRB_x_u}--\eqref{CRB_z_u} or in \eqref{matFcc_CRBx}--\eqref{matFcc_CRBz} for the unknown or known dipole orientation cases, respectively. We assume that the dipole is located at a distance $x_C = 6$\,m, the signal-to-noise ratio is $\mathrm{SNR}=|\chi|^2/\sigma^2=10$\,dB, and the wavelength is $\lambda =0.01$\,m (corresponding to $f_c=30$\,GHz), unless otherwise specified. {\color{orange} When needed to validate the asymptotic analysis, numerical results are given for surfaces of extremely large side length, e.g., up to $L = 10^3$\,m. Clearly, this does not mean that we advocate the use of such \emph{practically infinite} surfaces. Indeed, most of the conclusions and insights will be given for values in the range $1 \, {\rm m} \le L \le 10 \, {\rm m}$.}

\subsection{Analysis for the CPL case}

Fig.~\ref{fig:001_vs_size} shows the square root of the CRBs (RCRB), measured in meters $[\rm m]$, for the three components $x_C$, $y_C$ and $z_C$, as a function of the surface side length $L$, for a vertically oriented dipole located in CPL, i.e., under the hypotheses of Assumption~\ref{ass:cpl_vertical}. Both cases of unknown and known orientation are considered in Fig.~\ref{fig:unk001} and Fig.~\ref{fig:k001}, respectively. We see that all the RCRBs decrease fast with the surface side length, at least for values of $L$ of practical interest, i.e. in the range $1 \, {\rm m} \le L \le 10 \, {\rm m}$. \textcolor{blue}{The results in Fig.~\ref{fig:unk001} and Fig.~\ref{fig:k001} show that, for a vertically oriented dipole in CPL, the estimation accuracy is virtually the same with both known and unknown orientation - see also Fig.~\ref{fig:add_term}. We see that $\mathrm{RCRB}(x_C)$ is much lower than $\mathrm{RCRB}(y_C)$ and $\mathrm{RCRB}(z_C)$ in the range $1 \, {\rm m} \le L \le 10 \, {\rm m}$, and the asymptotic limit is achieved for $L \approx 20$ m. Also, as predicted by \eqref{eq:CRBYtoINF} and \eqref{eq:CRBZtoINF}, $\mathrm{RCRB}(y_C)$ and $\mathrm{RCRB}(z_C)$ decrease unboundedly as $L$ increases.} Notice that an accuracy on the order of tens of centimeters in all the three dimensions (as required for example in future automotive and industrial applications, e.g.,~\cite{Witrisal5Gaccuracy}) is achieved only for $L \approx 3 \, {\rm m}$, both with known or unknown orientation.

\begin{figure}[t!]
	\begin{overpic}[width=\columnwidth]{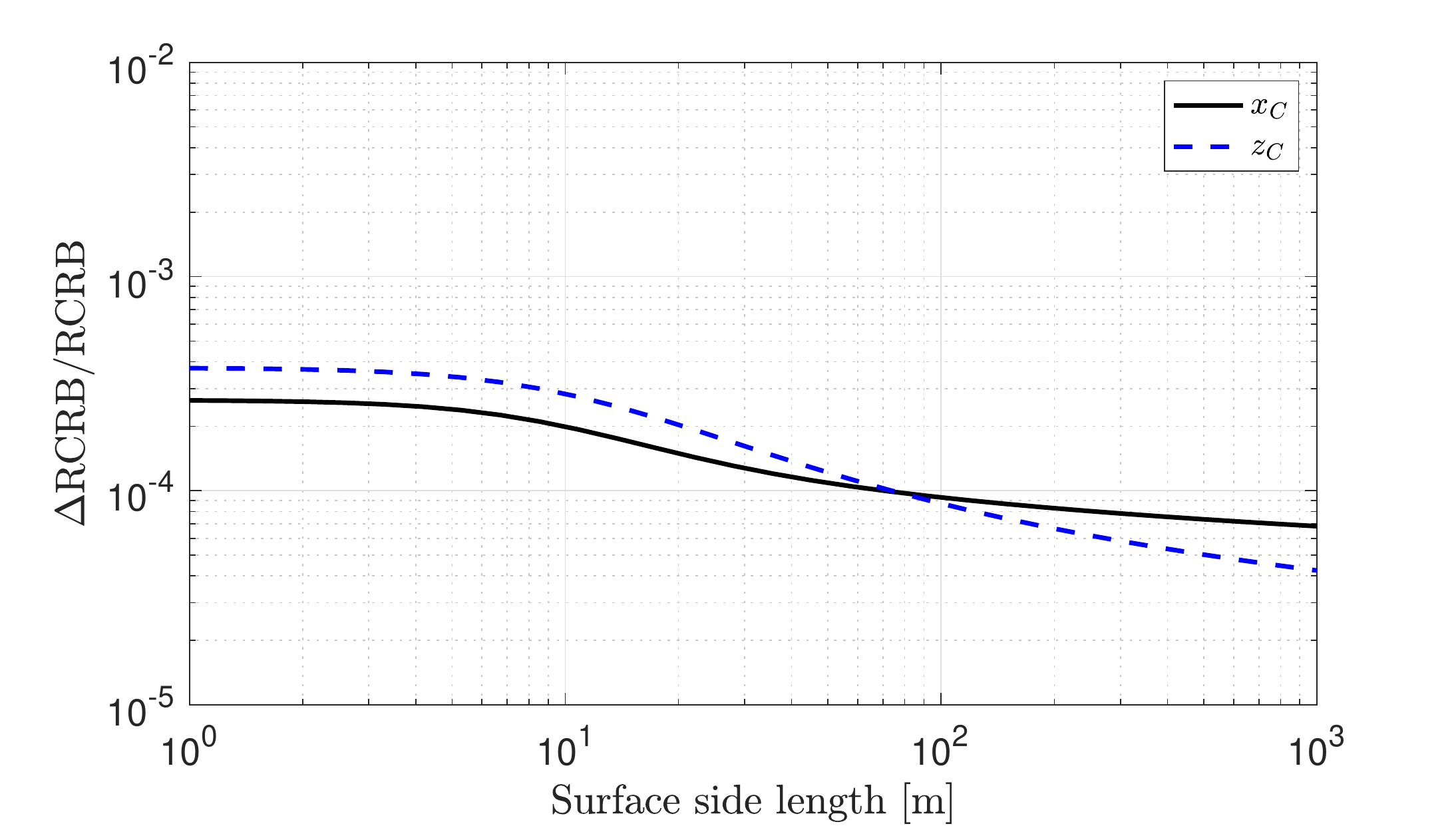}
	\end{overpic} 	

	\caption{Loss due to the lack of knowledge of dipole orientation}
	\label{fig:add_term}
\end{figure}

\begin{figure}[t!]
        \centering
	\begin{subfigure}[t]{1\columnwidth} \centering 
	\begin{overpic}[width=\columnwidth]{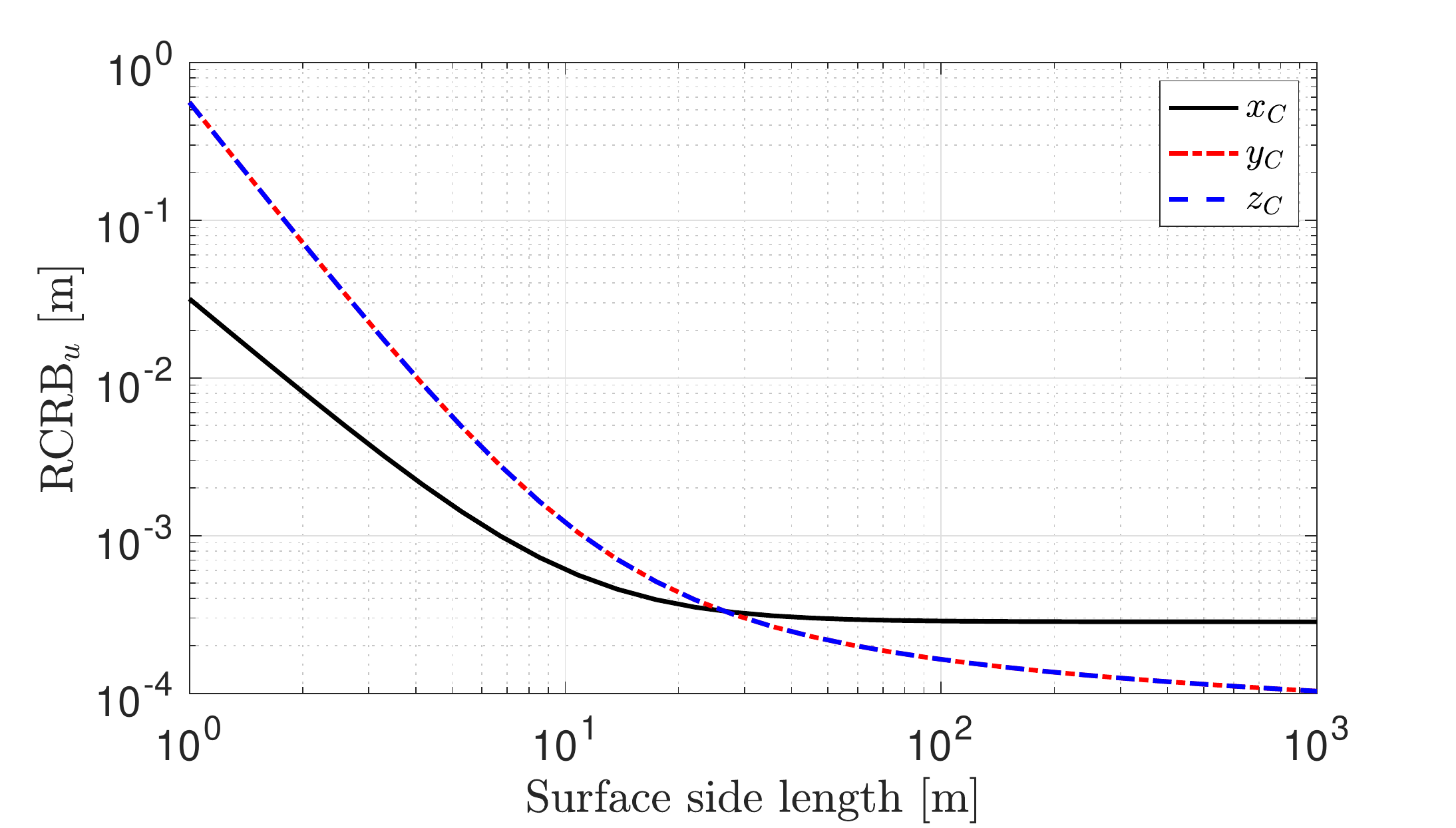}
	\end{overpic} 	

	\caption{Unknown dipole orientation, ${\bf \hat{t}}=(1,0,0)$}\vspace{0.2cm}
	\label{fig:unk100}
\end{subfigure}

\begin{subfigure}[t]{1\columnwidth} \centering 
	\begin{overpic}[width=\columnwidth]{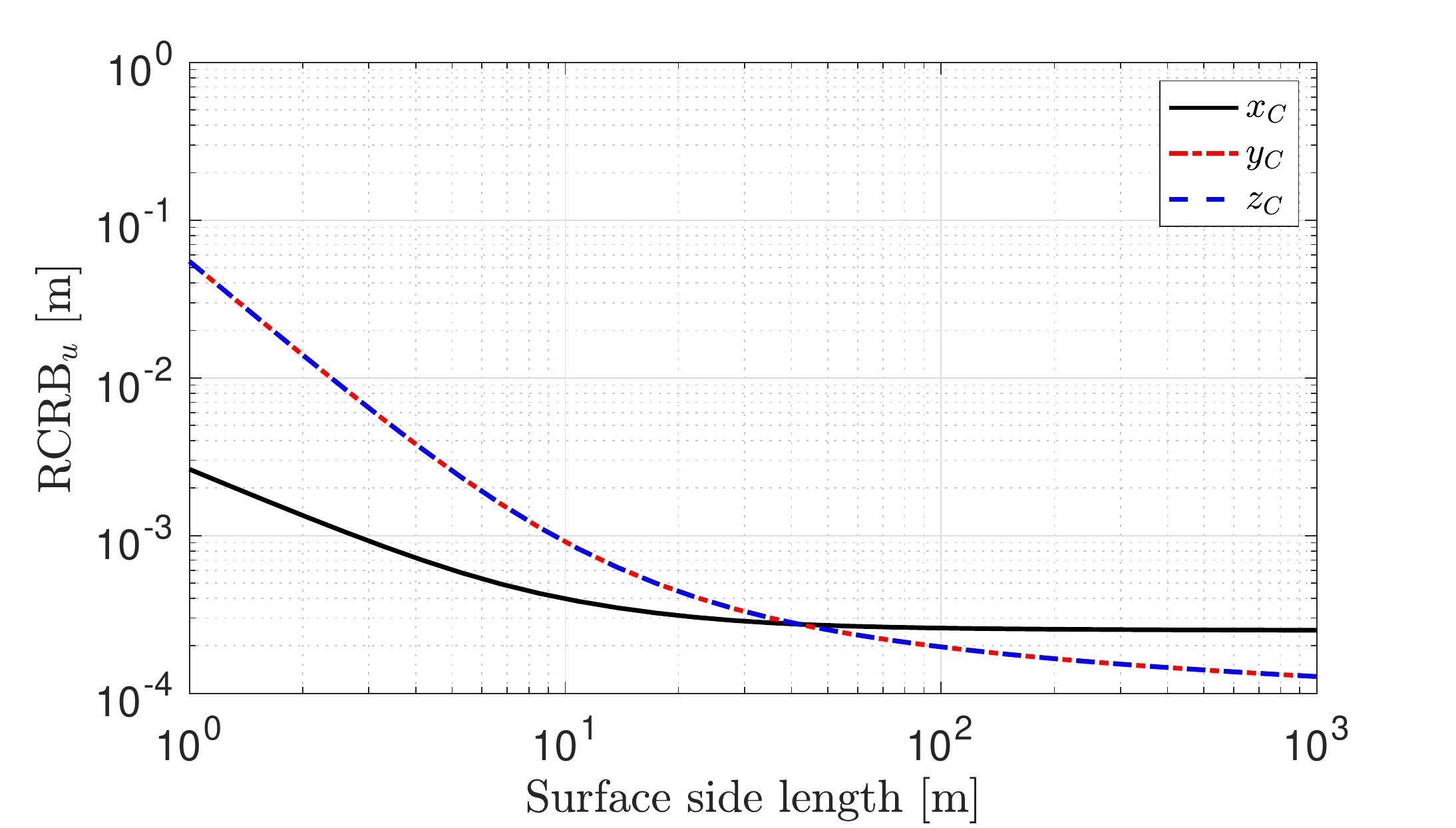}
	\end{overpic} 	

	\caption{Unknown dipole orientation, ${\bf \hat{t}}=(1,1,1)/\sqrt{3}$}
	\label{fig:unk111}
\end{subfigure}

\caption{RCRBs as a function of the surface side length for two different (unknown) dipole orientations.}
\end{figure}
{To quantify the difference between the CRBs with and without knowledge of dipole orientation, Fig.~\ref{fig:add_term} plots the quantity 
\begin{align}
\left[\dfrac{\Delta \rm{RCRB}}{\rm{RCRB}}\right]_{ii} \triangleq \sqrt{\dfrac{\left[\left(\mathbf{F}_{cc}-\mathbf{F}_{tc}^T\mathbf{F}_{tt}^{-1}\mathbf{F}_{tc}\right)^{-1}\right]_{ii}-\left[\mathbf{F}_{cc}^{-1}\right]_{ii}}{\left[\mathbf{F}_{cc}^{-1}\right]_{ii}}}
\end{align}
as obtained from \eqref{MIL}. Notice that $y_C$ is not reported as it is below the numerical precision of Matlab, as expected from the fact that, under Assumption~\ref{ass:cpl_vertical}, $\mathrm{CRB}_{\rm{u}}(y_C)=\mathrm{CRB}(y_C)$. We see that the loss, due to the lack of knowledge of dipole orientation, is negligible for all the considered values of $L$. {\color{orange}This conclusion does not hold true in general. In fact, the role of the dipole orientation in position estimation depends on the operating conditions, and should be verified on a case-by-case basis. In general, having or not knowledge of the dipole orientation has a negligible impact when one of the two (or both) conditions is satisfied: 1) the orientation is well estimated; 2) there is a weak interaction between orientation and position parameters, i.e., the FIM has a nearly block-diagonal structure. In the simulation setting of Figs. 3 and 4, the negligible impact is mainly due to a weak interaction between them.}}

\subsection{Impact of dipole orientation}

To quantify the impact of dipole orientation, Fig.~\ref{fig:unk100} and Fig.~\ref{fig:unk111} show the RCRBs for two different values of ${\bf \hat{t}}$, namely ${\bf \hat{t}}=(1,0,0)$ and ${\bf \hat{t}}=(1/\sqrt{3},1/\sqrt{3},1/\sqrt{3})$. In both cases, the dipole is in CPL and its orientation is unknown. Compared to Fig.~\ref{fig:unk001}, a sensible loss is evident only when ${\bf \hat{t}}=(1,0,0)$ and for small values of $L$. For example, with $L=3$ m the accuracy in the estimation of $x_C$ decreases from $10$ cm to $1$ m. Notice that, for both ${\bf \hat{t}}=(1,0,0)$ and ${\bf \hat{t}}=(1/\sqrt{3},1/\sqrt{3},1/\sqrt{3})$, the accuracy in the estimation of $y_C$ and $z_C$ is the same. Similar conclusions hold when the dipole orientation is known.

%\begin{figure}[th!]\centering
%	\begin{subfigure}[t]{.8\columnwidth} \centering 
%	\includegraphics[width=\columnwidth]{FigureArticolo/CoordinateX.eps}
%	\caption{$\sqrt{{\rm CRB}(x_C)}$ as a function of $y_C$ and $z_C$}
%	\label{fig:known_xc_vs_position}
%\end{subfigure}
%\\[12 pt]
%\begin{subfigure}[t]{.8\columnwidth} \centering 
%	\includegraphics[width=\columnwidth]{FigureArticolo/CoordinateY.eps}
%	\caption{$\sqrt{{\rm CRB}(y_C)}$ as a function of $y_C$ and $z_C$}
%	\label{fig:known_yc_vs_position}
%\end{subfigure}
%\\[12 pt]
%\begin{subfigure}[t]{.8\columnwidth} \centering 
%	\includegraphics[width=\columnwidth]{FigureArticolo/CoordinateZ.eps}
%	\caption{$\sqrt{{\rm CRB}(z_C)}$ as a function of $y_C$ and $z_C$}
%	\label{fig:known_zc_vs_position}
%\end{subfigure}
%\caption{RCRBs in the case of known orientation as a function of $y_C$ and $z_C$ when $x_C=6$ m. The dipole is oriented vertically and $L=3$ m.}
%\end{figure}

\begin{figure}[th!]\centering
	\begin{subfigure}[t]{0.8\columnwidth} \centering 
	\includegraphics[width=\columnwidth]{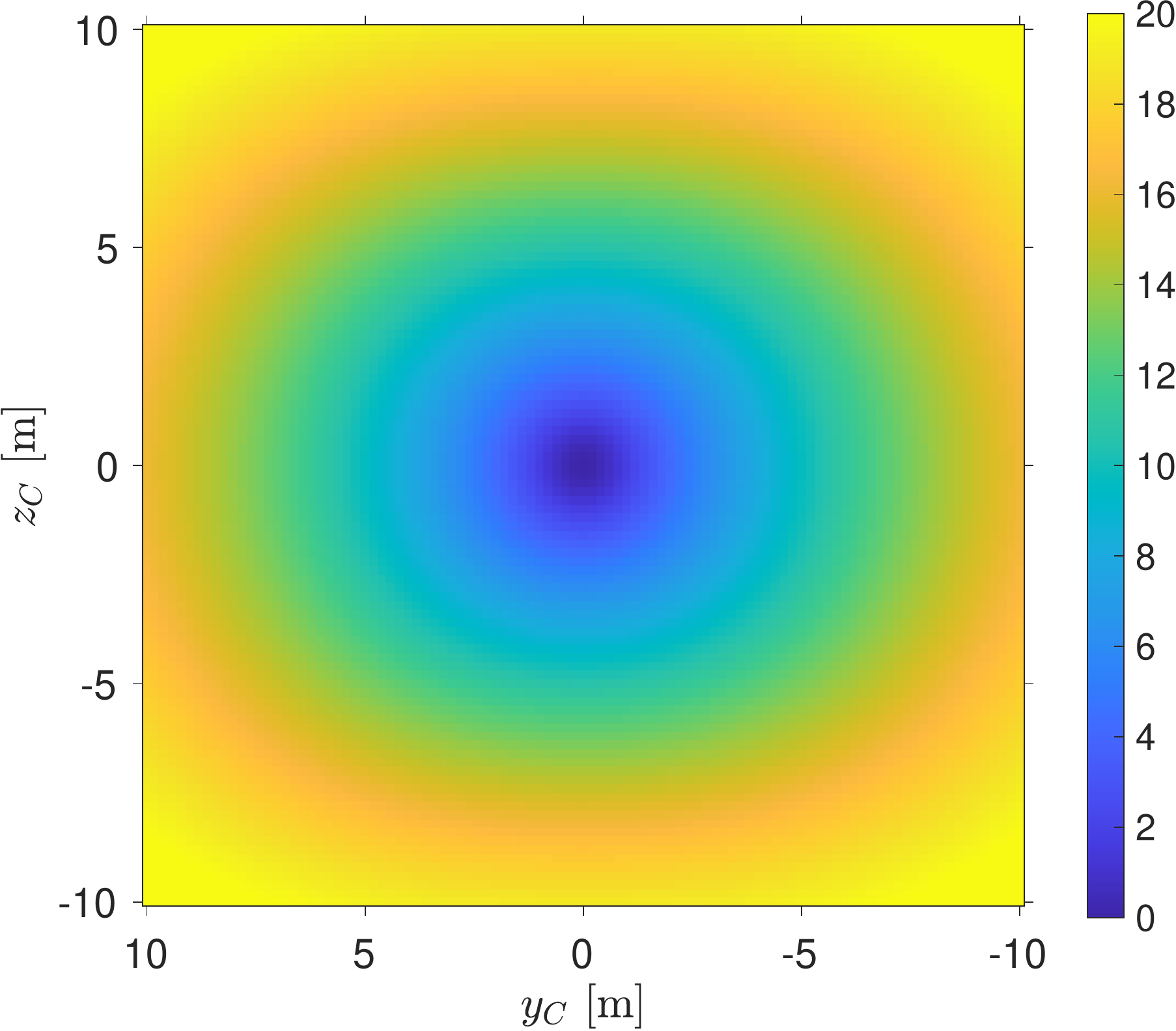}
	\caption{${\rm RCRB_u}(x_C)$ as a function of $y_C$ and $z_C$}
	\label{fig:unk_xc_vs_position}
\end{subfigure}
\\[12 pt]
\begin{subfigure}[t]{0.8\columnwidth} \centering 
	\includegraphics[width=\columnwidth]{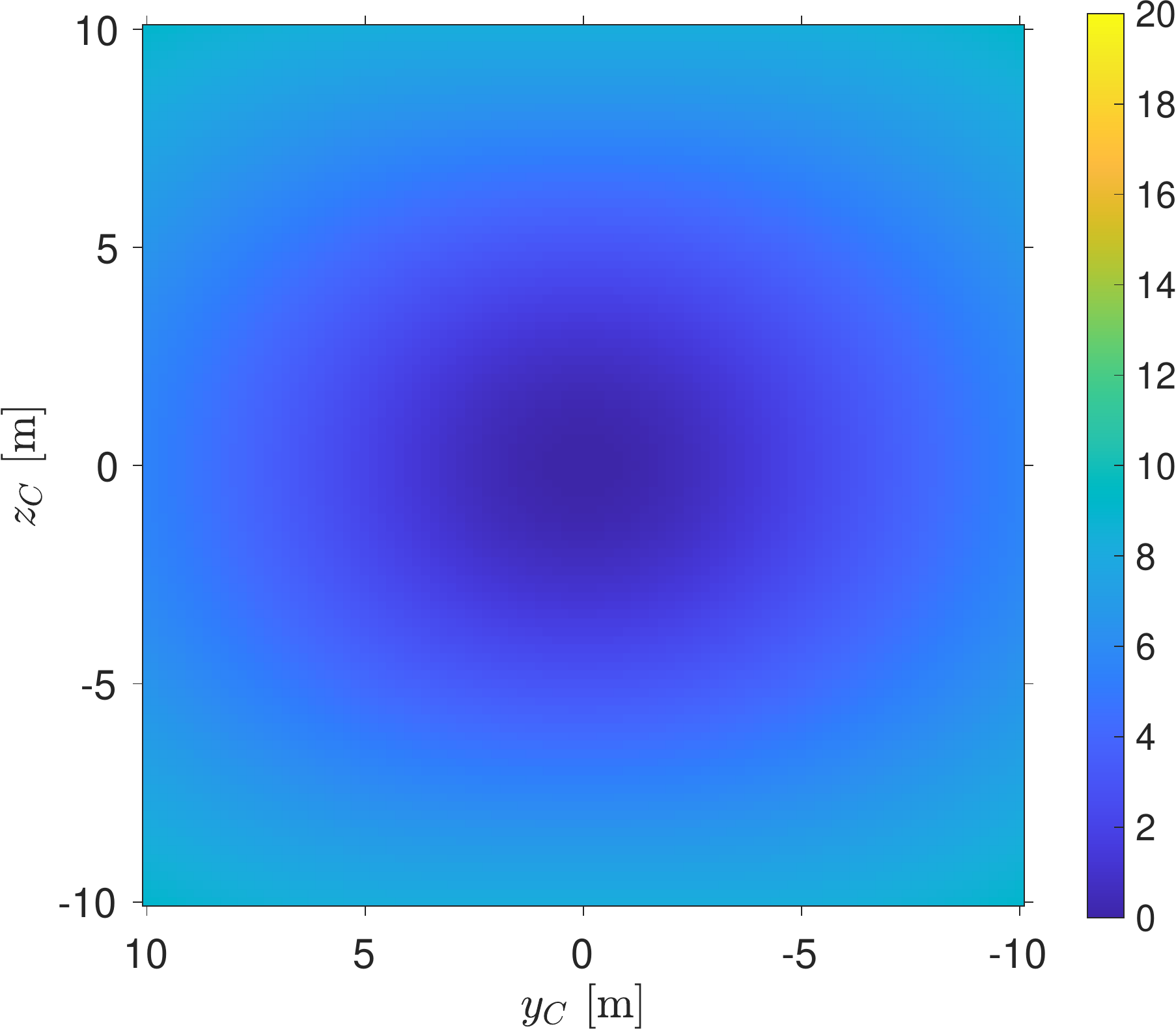}
	\caption{${\rm RCRB_u}(y_C)$ as a function of $y_C$ and $z_C$}
	\label{fig:unk_yc_vs_position}
\end{subfigure}
\\[12 pt]
\begin{subfigure}[t]{0.8\columnwidth} \centering 
	\includegraphics[width=\columnwidth]{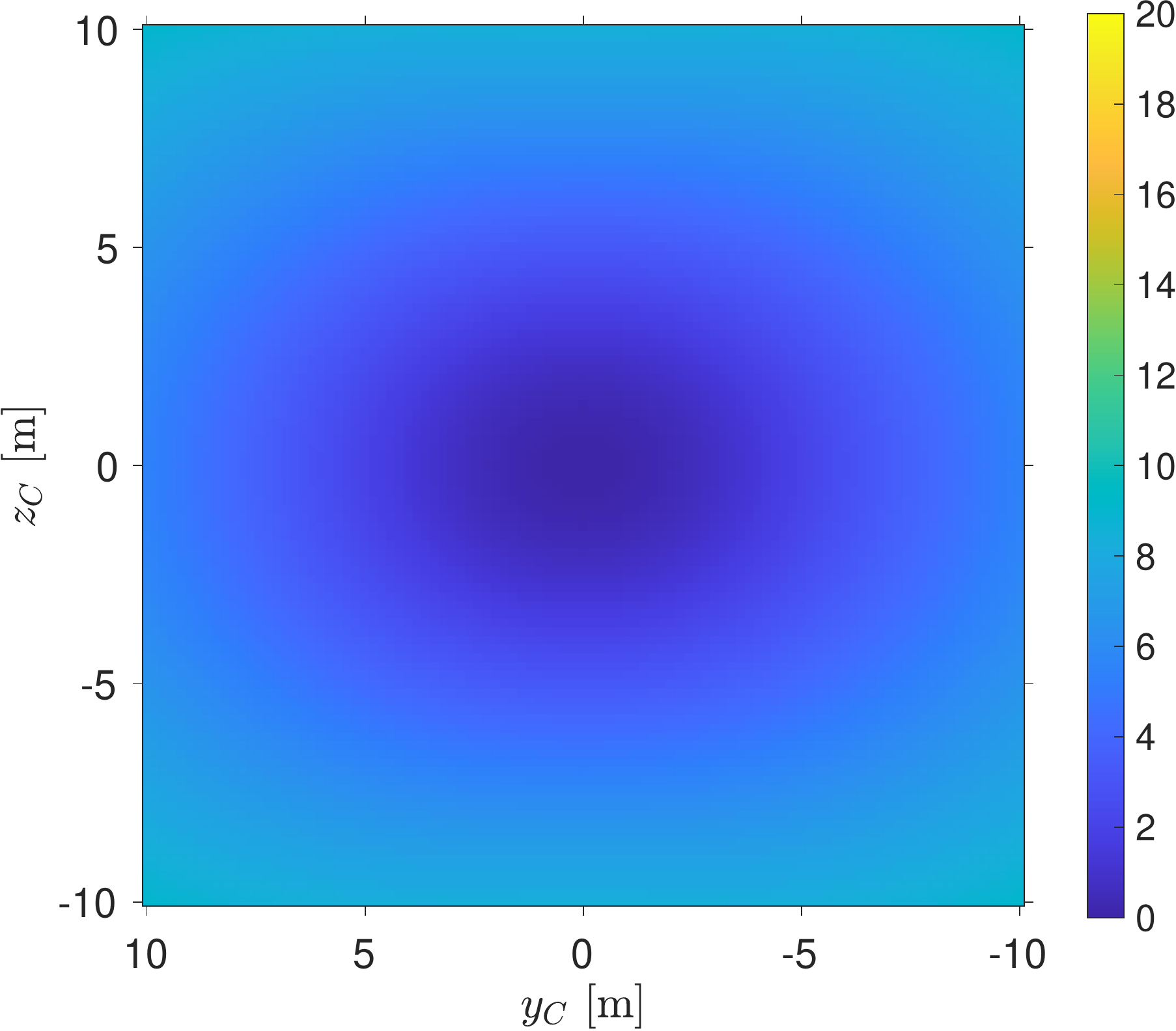}
	\caption{${\rm RCRB_u}(z_C)$ as a function of $y_C$ and $z_C$}
	\label{fig:unk_zc_vs_position}
\end{subfigure}
\caption{RCRBs in the case of unknown orientation as a function of $y_C$ and $z_C$ when $x_C=6$ m. The dipole is oriented vertically and $L=3$ m.}
\label{fig:unk_vs_position}
\end{figure}

\textcolor{blue}{Fig.~\ref{fig:unk_vs_position} shows the RCRBs in the case of unknown orientation, respectively, as a function of $y_C$ and $z_C$ when $x_C=6$ m. The dipole is oriented vertically and $L=3$ m. The value of the RCRB corresponding to a point $(y_C,z_C)$ is measured by the color of that point. More precisely, the RCRB values are first normalized to their minimum, which is achieved when the dipole is in CPL ($y_C=z_C=0$), and then the normalized values (in dB) are mapped into a colour: higher values are associated to warm colours, lower values to cool ones. This means, for example, that the blue zones in each figure correspond to the best estimation accuracy. Fig.~\ref{fig:unk_vs_position} clearly shows the different behaviours of various RCRBs when the dipole moves away from the CPL position. For example, we see that ${\rm RCRB_u}(x_C)$ increases faster than ${\rm RCRB_u}(y_C)$ and ${\rm RCRB_u}(z_C)$ whatever the direction of motion is. On the other hand, ${\rm RCRB_u}(y_C)$ and ${\rm RCRB_u}(z_C)$ have a similar behavior. The same conclusions hold when the dipole orientation is known.}

\blue{From the above analyses, we conclude that the CRBs are approximately the same whether or not the orientation of the transmitting dipole is known. However, this does not necessarily mean that we can ignore the effect of dipole orientation in the estimation process. It only means that the joint estimation of $\bf t$ and $\bf u$ ultimately provides the same localization accuracy as if $\bf t$ were known. The impact of the orientation knowledge will be quantified in Section VI where MLEs will be considered.}

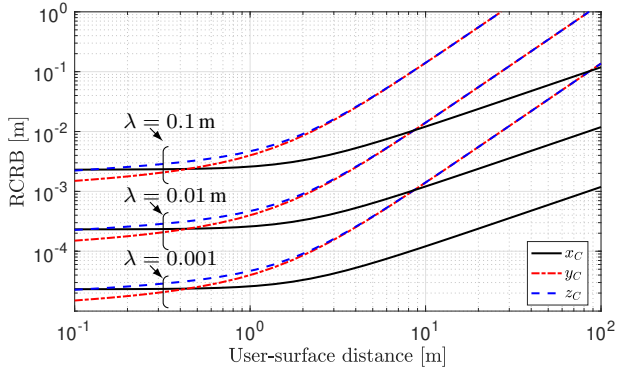
\begin{figure}[t!]\centering

	\includegraphics[width=1.02\columnwidth]{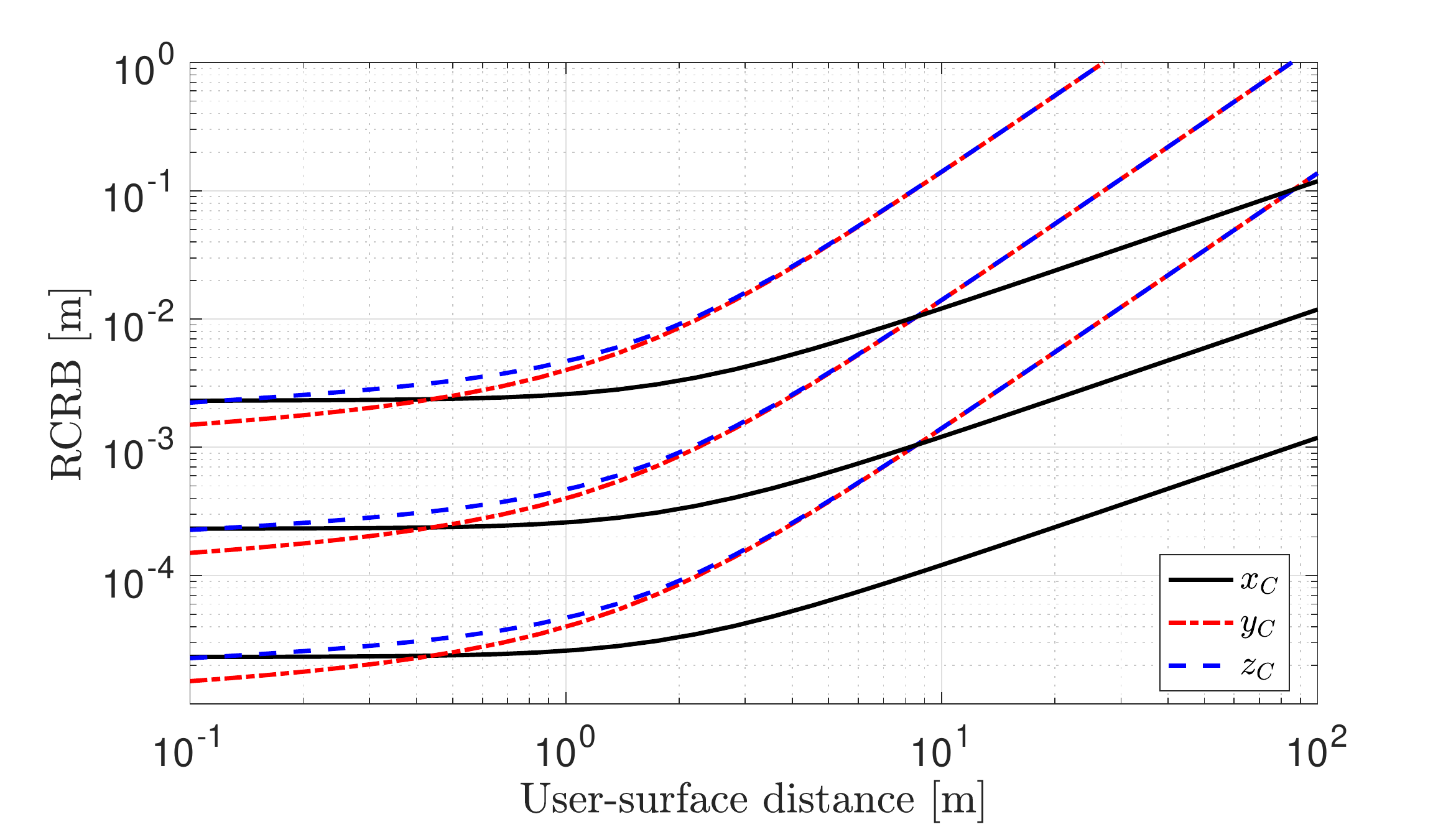}
		\put(-205,39){\footnotesize $\lambda =0.001$}
	 \put(-195,39){\vector(1, -1){5}}
	 	\put(-188,29){\oval(4,12)[l]}
	\put(-205,63){\footnotesize $\lambda =0.01$\,m}
	\put(-195,63){\vector(1, -1){5}}
	\put(-188,52){\oval(4,14)[l]}
		\put(-205,91){\footnotesize $\lambda =0.1$\,m}
	\put(-195,91){\vector(1, -1){5}}
	\put(-188,77){\oval(4,14)[l]}\caption{RCRBs as a function of $x_C$ for a vertical dipole in CPL and $L = 3$ m, with $\lambda = 0.1,~0.01$ or $0.001$\,m.}	\label{fig:dist_A}
\end{figure}

\subsection{Impact of carrier frequency}

\textcolor{blue}{Fig.~\ref{fig:dist_A} shows the RCRBs as a function of $x_C$ for three different values of the wavelength, namely $\lambda=0.1$ m (corresponding to $f_c=3$\,GHz), $\lambda=0.01$ m (corresponding to $f_c=30$\,GHz) and $\lambda=0.001$ m (corresponding to $f_c=300$\,GHz). The dipole is in CPL and is oriented vertically. Its orientation is unknown to the receiver. The surface side length is $L=3$ m.  As expected, the estimation accuracy reduces as the distance between the source and the observation region increases. In particular, we see that ${\rm RCRB}(x_C)$ increases \textit{slower} than ${\rm RCRB}(y_C)$ and ${\rm RCRB}(z_C)$. Notice that the RCBRs depend linearly on $\lambda$, at least in the range of values of $x_C$ (and hence $\rho$) considered in Fig.~\ref{fig:dist_A}. Indeed, reducing the wavelength by a factor of $10$ reduces the RCBRs of the same factor. This can easily be derived analytically for $x_C$ and $y_C$ by considering the results in Proposition~\ref{prop:xcgglambda}. This holds true also for $z_C$, as it is shown in Fig.~\ref{fig:dist_A}. Notice that the same dependence on $\lambda$ was already observed in \cite{Hu2018b}. Similar results (not shown due to space limitations) can be obtained if the dipole orientation is known. Marginal differences are only observed for $\lambda =0.1$\,m and small values of $L$, i.e., $L\le 0.1$\,m.}

\section{Evaluating the impact of different models}
\textcolor{blue}{To showcase how different electric field models impact the estimation accuracy, we now use the derived CRBs to benchmark different MLEs derived on the basis of a discrete representation of misspecified models. Particularly, we consider a practical scenario in which the observation region is filled with vertically-oriented short dipoles. The analysis is carried out under the hypotheses of Assumption~\ref{ass:Frau} for $l_s = \lambda/4$ and $\lambda =0.1$. Since $2 l_s^2 /\lambda = \lambda/8 = 0.0125$\,m, this means that it is valid for all distances of practical relevance.}

 \begin{remark}
 {\color{orange} When the assumed model differs from the true one, the estimation problem is said to be misspecified or mismatched~\cite{white1982maximum,richmond2015parameter}. In these circumstances, fundamental limits can be computed by resorting to the mismatched estimation theory. which allows to derive the CRBs under model mismatching. This is without any doubt an interesting extension of our theoretical analysis, which is left for future work. A comprehensive review on the subject can be found in~\cite{fortunati2017performance}.}
\end{remark}
%\textcolor{blue}{n order to answer this question, we start from the \textit{exact} expression of the electric field produced by a Hertzian dipole, which is obtained by simply substituting \eqref{DyadicGF1} and \eqref{J_dipole} into \eqref{e1}. This yields~\cite[\S 15.5]{OrfanidisBook}
%\begin{equation}
%\label{TruE}
%\bar{\bf e}({\bf r})={\bf e}({\bf r})+{\bf e}_{\rm nf}({\bf r})
%\end{equation}
%with ${\bf e}({\bf r})$ given in \eqref{E_P_PTAI} and
%\begin{equation}
%\label{enf}
%{\bf e}_{\rm nf}({\bf r})=I_{in}l_{s}G(r)\dfrac{1}{kr}\left(1-\imagunit \dfrac{1}{kr} \right)[\mathbf{\hat{t}}-3(\mathbf{\hat{r}}\cdot \mathbf{\hat{t}}) \mathbf{\hat{r}}].
%\end{equation}
%} 
%\blue{It can easily be shown that $\bar{\bf e}({\bf r})$ reduces to ${\bf e}({\bf r})$ when $kr_{o} \gg 1$, which is met in practical applications.}
%\blue{In general, what happens when \eqref{TruE} is replaced with approximate models can only be established by computing \textit{misspecified} CRBs \cite{Richmond2015}. However, with the purpose of giving a partial (but significant) answer to this question, in the next section we analyze the performance of practical maximum-likelihood estimators derived under misspecified models.}

\vspace{-0.7cm}
\textcolor{blue}{\subsection{Discrete signal model}}
\textcolor{blue}{We assume that the observation region is filled with short dipoles of length $l_r = \lambda/10$, vertically oriented and placed on a square grid. The centers of the dipoles are the set of points of $\mathcal{R}_o$ given by $\{(x,y,z): x=0, y=m \lambda/{2}, z=n \lambda/{2}\}$, with $1 \le |m|,|n|  \le N_r$ and $N_r=\lfloor L/\lambda \rfloor$. The voltage $V_{mn}$ observed at the output of the $(m,n)$ receive dipole is obtained by integrating over the antenna length the vertical component given by
\begin{equation}
\xi_z({\bf r})={e}_{z}({\bf r}) + n_{z}({\bf r}).
\end{equation}
Since the Hertzian dipole is electrically small, i.e., $l_r \ll \lambda$, it follows that $V_{mn}$ can be approximated as
\begin{equation}
\label{V_mn}
V_{mn}=\int_{l_r} \xi_z({\bf r}) dz \approx h_{mn}+\nu_{mn}
\end{equation}
where
\begin{align}
h_{mn}=l_{r}e_{z}({\bf r}_{mn})
\label{hmn}
\end{align}
with ${\bf r}_{mn}=x_C {\bf{\hat x}}+(y_m-y_C){\bf{\hat y}}+(z_n-z_C){\bf{\hat z}}$, and $\{\nu_{mn}\}$ are independent zero-mean gaussian random variables, with variance {$\sigma^2_{\nu}=2\sigma^2 l_r/\lambda$}.}

\textcolor{blue}{\subsection{Maximum-likelihood estimation under \textit{misspecified} models}}

\blue{The log-likelihood function for the estimation of ${\bf \hat{t}}=(t_x,t_y,t_z)$ and ${\bf u}=(x_C,y_C,z_C)$ on the basis of the observations $\{V_{mn};|m|,|n|=1,\ldots,N_r\}$ is given by~\cite[Ch. 7]{Kay1993a}
\begin{equation}
\label{LLF}
\Lambda(\tilde{\bf t},\tilde {\bf u})=-\underset{|m|,|n|=1,\ldots,N_r}{\sum\sum} |V_{mn}-\tilde h_{mn}|^2
\end{equation}
where $\tilde{\bf t}=(\tilde t_x,\tilde t_y,\tilde t_z)$ and $\tilde {\bf u}=(\tilde x_C,\tilde y_C,\tilde z_C)$ are trial values for $\bf{\hat t}$ and $\bf u$, respectively, and $\tilde h_{mn}$ is obtained accordingly. Different MLEs can be obtained if different models are assumed for $\tilde h_{mn}$ in \eqref{LLF}. Specifically, we consider the following three.
%The MLE \eqref{MLE} is based on the received signal model given by \eqref{hmn}. Simplified MLEs can be derived starting from approximate models. We consider the three MLEs obtained by assuming the following \textit{mismatched} expressions for $h_{mn}$: 
\begin{enumerate}
  \item The first MLE (MLE1) relies on the model provided in~\eqref{E_P_PTAI} and assumes that:
  \begin{equation} 
\label{hmn1}
h_{mn}^{(1)}=l_{r} e_{z}({\bf r}_{mn}).
\end{equation}
  \item The second MLE (MLE2) makes use of the signal model adopted in~\cite{Hu2018b}. Hence, we have that:
  \begin{equation}
\label{hmn2}
h_{mn}^{(2)}=l_{r}G(r) \sqrt{\dfrac{x_{C}}{r}}.
\end{equation}
  \item The third MLE (MLE3) is based on the standard planar approximation~\eqref{Gr3} under which:
  \begin{equation}
\label{hmn3}
h_{mn}^{(3)}=l_{r}G(r_{C})e^{-\imagunit k ({\bf \hat r}_C \cdot {\bf d}_{mn})}
\end{equation} with ${\bf d}_{mn}=y_m{\bf{\hat y}}+z_n{\bf{\hat z}}$. \textcolor{orange}{Notice that with the standard planar wave model, the amplitude of the received signal is proportional to $1/r_{C}$. Such a dependence can effectively be exploited for the estimation of $r_{C}$ provided that the proportionality factor is exactly known.}
\end{enumerate}}
\blue{The first MLE takes the form
\begin{equation}
\label{MLEapp}
{\bf{\hat t}}^{(1)},{\bf u}^{(1)}=\arg \underset{\tilde{\bf t},\tilde{\bf u}}{\max} \, \Lambda^{(1)}(\tilde{\bf t},\tilde {\bf u})
\end{equation}
where
\begin{equation}
\label{LLFapp}
    \Lambda^{(1)}(\tilde{\bf t},\tilde {\bf u})=-\underset{|m|,|n|=1,\ldots,N_r}{\sum\sum} |V_{mn}-\tilde h^{(1)}_{mn}|^2.\end{equation}
In case of known orientation, the maximization of \eqref{MLEapp} is carried out only over $\tilde {\bf u}$, after replacing $\tilde{\bf t}$ with the true value ${\bf \hat {t}}$ in $\tilde h^{(1)}_{mn}$.}

\blue{The MLEs derived from \eqref{hmn2}--\eqref{hmn3} are in the following form:
\begin{equation}
\label{MLEapp23}
{\bf u}^{(i)}=\arg \underset{\tilde{\bf u}}{\max} \, \Lambda^{(i)}(\tilde {\bf u})
\end{equation}
where
\begin{equation}
\label{LLFapp23}
    \Lambda^{(i)}(\tilde {\bf u})=-\underset{|m|,|n|=1,\ldots,N_r}{\sum\sum} |V_{mn}-\tilde h^{(i)}_{mn}|^2\end{equation}
 for $i=2,3$.} \blue{Notice that \eqref{MLEapp23} needs to be optimized only with respect to $\tilde {\bf u}$. This is because $h_{mn}^{(2)}$ and $h_{mn}^{(3)}$ do not account for the dependence of the received signal on the orientation of the source dipole. This will have a profound impact on the estimators' performance.} 
\begin{figure}[t!]\centering
	\begin{subfigure}[t]{1\columnwidth} \centering 
	\includegraphics[width=\columnwidth]{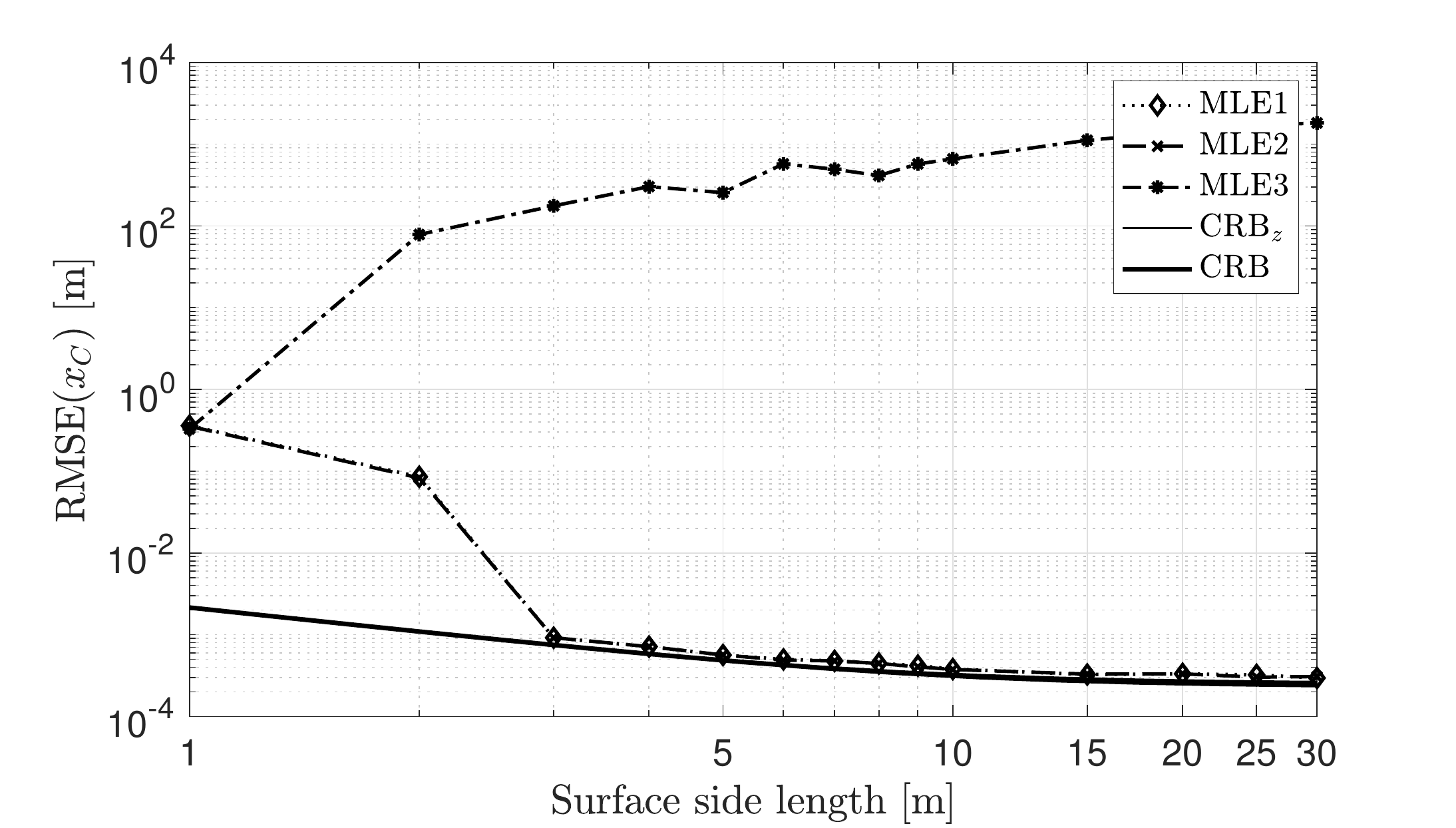}
	\caption{Estimation accuracy for $x_C$}
	\label{fig:MLaccuracy_x_c}
\end{subfigure}

	\begin{subfigure}[t]{1\columnwidth} \centering 
	\includegraphics[width=\columnwidth]{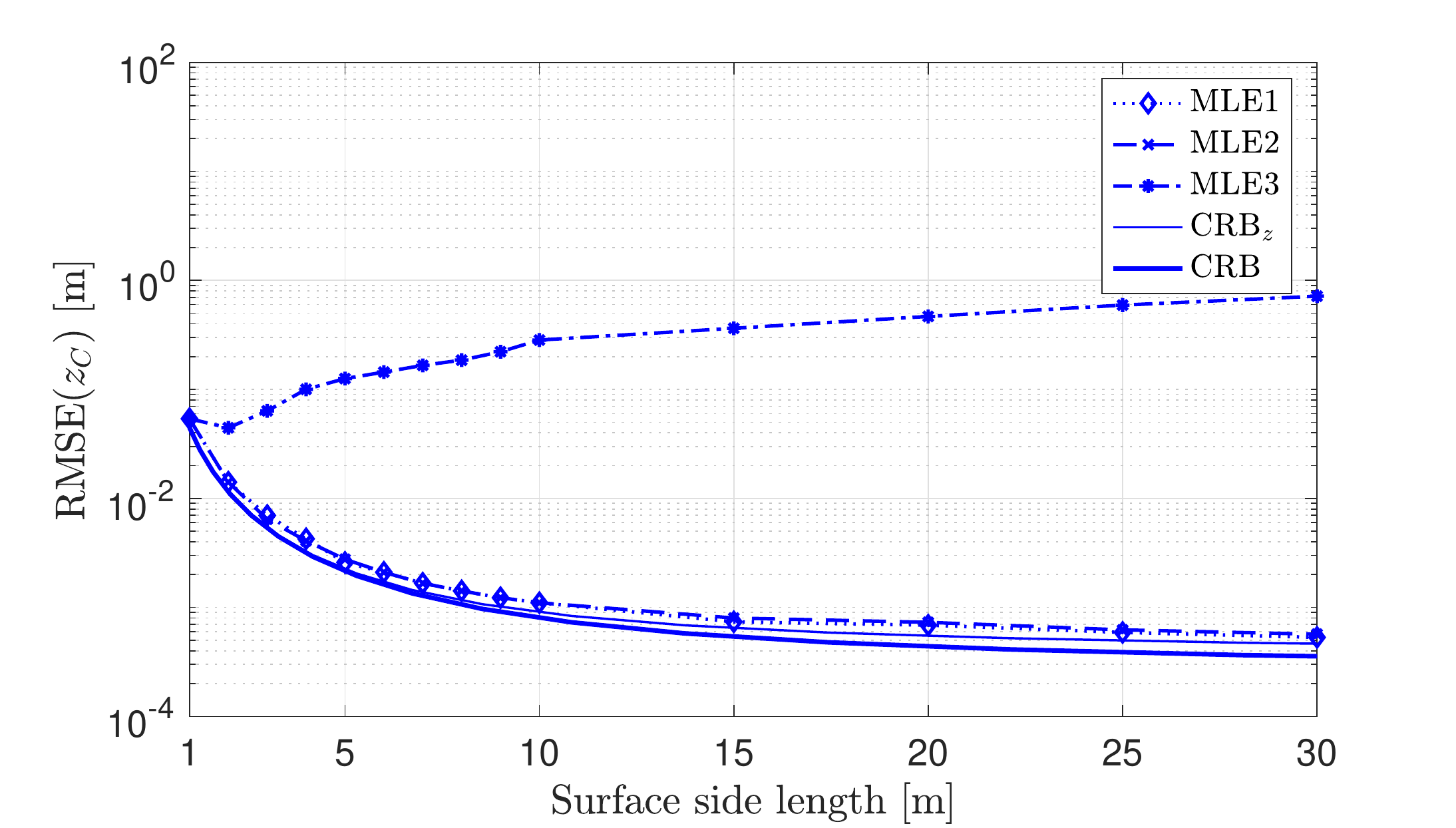}
	\caption{Estimation accuracy for $z_C$}
	\label{fig:MLaccuracy_z_c}
\end{subfigure}

\caption{Comparisons between MLEs when ${\bf \hat {t}}=(0,0,1)$.}
\label{RCRBvsDistance}
\end{figure}

\begin{figure}[t!]\centering
	\begin{subfigure}[t]{1\columnwidth} \centering 
	\includegraphics[width=\columnwidth]{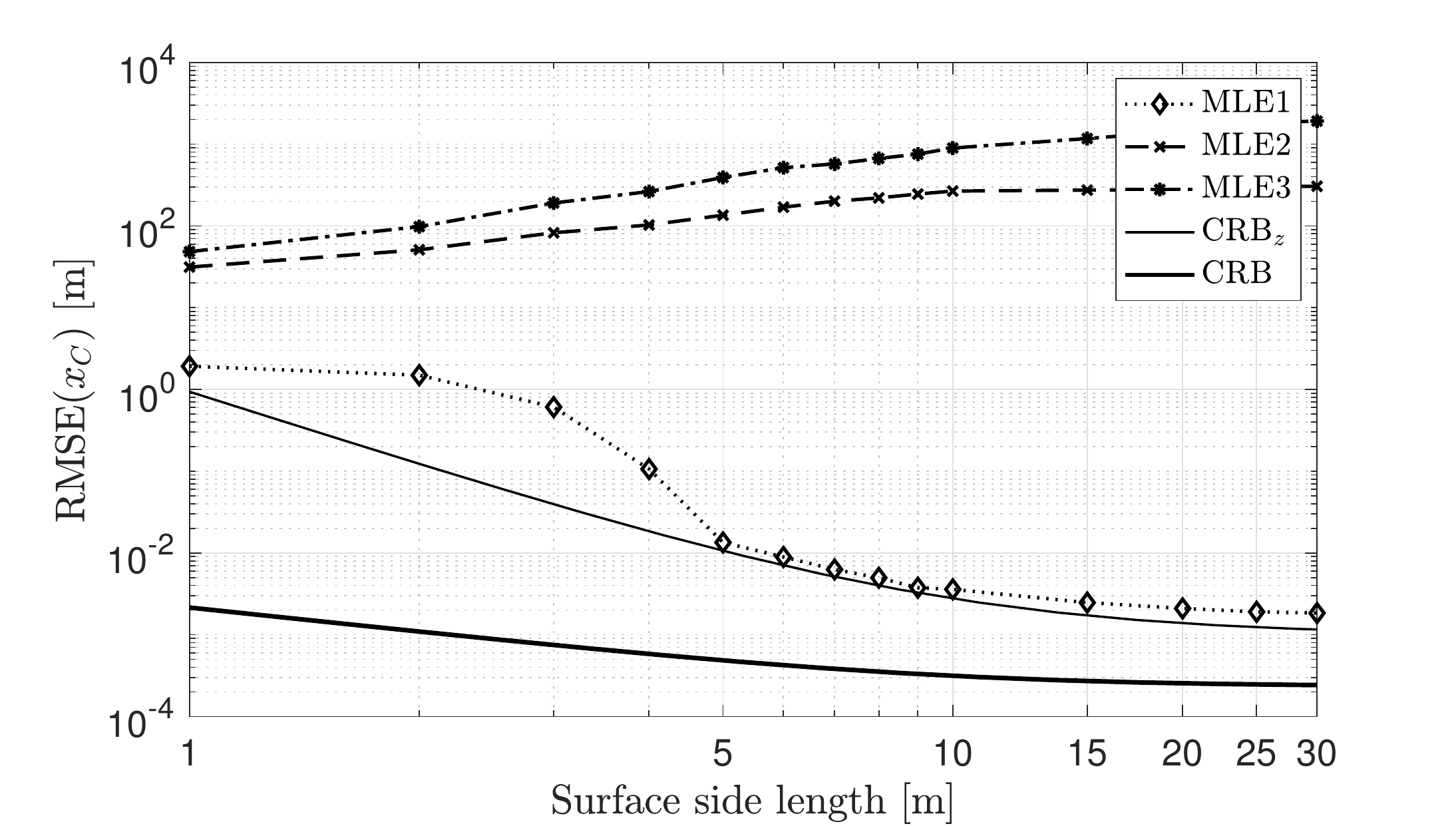}
	\caption{Estimation accuracy for $x_C$}
	\label{fig:MLaccuracy_x_c_010}
\end{subfigure}

	\begin{subfigure}[t]{1\columnwidth} \centering 
	\includegraphics[width=\columnwidth]{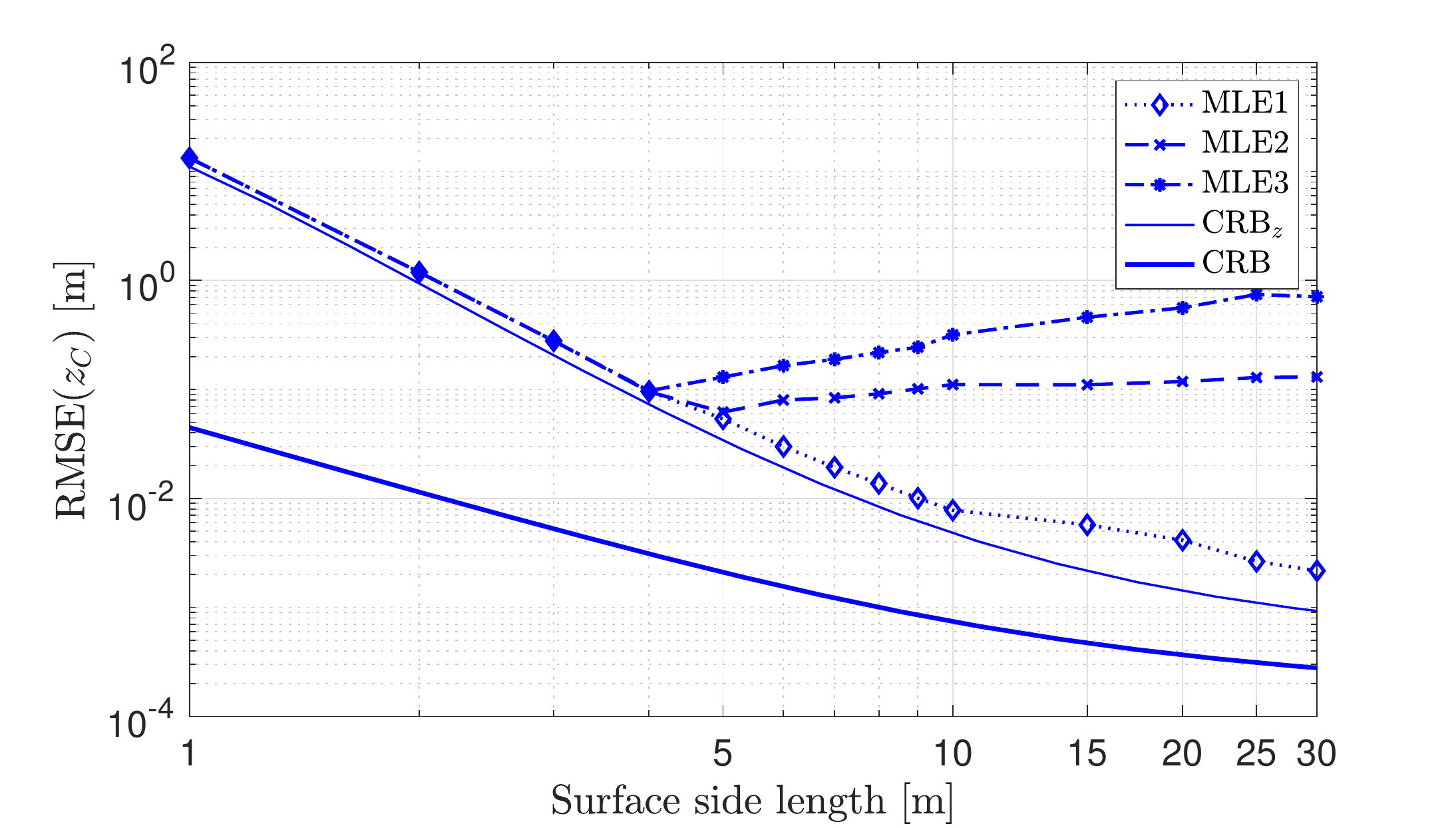}
	\caption{Estimation accuracy for $z_C$}
	\label{fig:MLaccuracy_z_c_010}
\end{subfigure}
\caption{Comparisons between MLEs when ${\bf \hat {t}}=(0,1,0)$.}
\label{RCRBvsDistance}
\end{figure}

\subsection{Numerical analysis}
\textcolor{blue}{The accuracy of the above three MLEs is now quantified in terms of the square root of the MSE (RMSE), measured in meters and given by
\begin{equation}
\label{RMSE}
{\rm RMSE}(\upsilon)=\sqrt{{\rm E}(\upsilon-\hat \upsilon)^{2}}
\end{equation}
where $\upsilon$ stands for $x_{C}$, $y_{C}$ or $z_{C}$, and $\hat \upsilon$ is its ML estimate. 
%The curves corresponding to MLE1 in \eqref{MLEapp} are labeled with ${\rm RMSE}^{(1)}(\upsilon)$ and ${\rm RMSE}^{(1)}_{u}(\upsilon)$ for the cases of known or unknown dipole orientation, respectively, while ${\rm RMSE}^{(2)}(\upsilon)$ and ${\rm RMSE}^{(3)}(\upsilon)$ refer to the mispecified MLEs given in~\eqref{MLEapp23}--\eqref{LLFapp23}. 
The estimation accuracy of the three MLEs is compared with the provided CRBs. Since $V_{mn}$ is obtained by integrating $\xi_z$, we also report CRB$_{z}$, which is derived on the observation of the only $z$-component of the electric field.} 

\blue{Fig.~\ref{fig:MLaccuracy_x_c} shows ${\rm RMSE}(x_{C})$ as a function of $L$, measured in meters. The transmitting dipole is in CPL, with $x_{C}=6$ m, and is vertically oriented, i.e., ${\bf \hat {t}}=(0,0,1)$. The wavelength is $\lambda=0.1$ m, and ${\rm SNR}=30$ dB.  The following interesting conclusions can be drawn from the results in Fig.~\ref{fig:MLaccuracy_x_c}:}
\begin{itemize}
  \item A negligible difference is observed between ${\rm CRB}(x_{C})$ and ${\rm CRB}_{z}(x_{C})$. This means that, when the transmit and receive dipoles have the same orientation, the component of the electric field along the dipoles' orientation provides almost all the information about the source position.
  \item  MLE1 performs closely to the CRB in both cases of known and unknown orientation. This confirms the results from Section~\ref{Sec:Numerical_analysis}, in which it has been shown that $\rm CRB$ and ${\rm CRB}_{u}$ are practically the same.
  \item When the transmit dipole is vertically oriented, the estimation accuracy achieved with MLE2 is the same as that with MLE1.
  \item MLE3 has poor performance and its estimation accuracy worsens as $L$ grows.  \textcolor{orange}{This is due to the fact that the underlying planar model in \eqref{hmn3} becomes more and more inaccurate as the surface side length gets larger.}
\end{itemize}
%
%{Last estimator, MLE3, results in an opposite behaviour compared to the others. That is mainly due to the planar wave model, where the wave modulus remain constant over the whole surface. Conversely, the proposed model modulus decreases when with the distance form the source. The difference between these two models leads to an unbounded RMSE.}

\blue{The same observations are valid for $y_{C}$ and $z_{C}$. For example, Fig.~\ref{fig:MLaccuracy_z_c} shows ${\rm RMSE}(z_{C})$ in the same setting of Fig.~\ref{fig:MLaccuracy_x_c}. We see that MLE1 and MLE2 have similar behaviors, and are very close to ${\rm CRB}_{z}(z_{C})$, while MLE3 is far from the bound. A slight difference is observed between ${\rm CRB}(z_{C})$ and ${\rm CRB}_{z}(z_{C})$, for $L > 5$ m.}

\blue{Fig.~\ref{fig:MLaccuracy_x_c_010} illustrates ${\rm RMSE}(x_{C})$ assuming that the transmitting dipole is oriented along the $y$-direction, i.e., ${\bf \hat {t}}=(0,1,0)$. The other parameters are the same as in Fig.~\ref{fig:MLaccuracy_x_c}. Compared to Fig.~\ref{fig:MLaccuracy_x_c}, we observe two fundamental differences. First, ${\rm CRB}(x_{C})$ and ${\rm CRB}_{z}(x_{C})$ differ significantly. This means that, in general, to achieve the best estimation accuracy all the components of the electric field should be considered. Secondly, MLE2 has poor performance, because it ignores the dependence of the received signal on the orientation of the transmit dipole. This confirms the importance of the term $ \left[{\bf \hat r} {\bf \times} {\bf R}(\theta,\phi) \right] {\bf \times} {\bf \hat r}$ in \eqref{E_P.0}, which is not considered in the standard models adopted in literature. \textcolor{orange}{Thirdly, the estimation accuracy with both MLE2 and MLE3 get worse as the surface side length gets larger due to the increasing inaccuracy of the underlying  wave models in \eqref{hmn2} and \eqref{hmn3}, respectively.} Similar conclusions can be drawn from the plots in Fig.~\ref{fig:MLaccuracy_z_c_010}, which show the accuracy in the estimation of $z_{C}$. In this case, we see that the performance of MLE1 differs from ${\rm CRB}_{z}(x_{C})$ only for $L>10$ m, where the estimation accuracy is anyway better than $1$ cm.
}

%% Start algorithms figures

%% End algorithms figures

\section{Conclusions}\label{Sec:conclusions}
Large antenna arrays and high frequencies opens up opportunities for new signal processing algorithms for positioning. {\color{orange}Motivated by the need of establishing ultimate bounds, we provided a general model for the electric vector field over a spatially-continuous rectangular region.} Unlike standard models in signal processing, the functional dependence on the radiation vector at the source \blue{is intrinsically captured by the model}. The electric vector field model was used to compute the CRBs for the three-dimensional (3D) spatial location of a Hertzian dipole, with and without a priori knowledge of its orientation. Further simplifications and insights were obtained by assuming that the dipole center is located on the line perpendicular to the surface center. Numerical results showed that a centimeter-level accuracy {\color{orange}in the mmWave and sub-THz bands can only be achieved with surfaces of size in the range of a few meters.} Asymptotic expressions were also given in closed-form to show the scaling behaviors with respect to surface area and wavelength. \blue{To show how different electric field models impact the estimation accuracy, we used the derived CRBs to benchmark different MLEs. The analysis showed that the standard model neglecting the radiation angular pattern of the transmitting source may provide low estimation accuracy in the presence of large surfaces.}

{\color{orange}The statements above does not want to mean that the standard model is not useful. Only that its limitations need to be clearly acknowledged and understood and its performance carefully evaluated on a case-by-case basis. The standard models may be sufficiently accurate in some situations and inadequate in others. It is essential to evaluate the performance of algorithms based on the standard model against data generated using the analytic model.}

%The ultimate goal of positioning is to precisely estimate not only the 3D spatial location, but also the 3D orientation of the source. This requires the computation of the CRB with no a priori knowledge of the orientation, which will be addressed in the extended journal version.

 %!TEX root = main.tex
\section*{Appendix A}

\label{app:derCRB_dipole}

We compute the derivatives of the cartesian components of the electric field with respect to the unknown parameters $(t_x,t_y,t_z,x_C,y_C,z_C)$. Starting from \eqref{CartesianComponents_PTAI_x}--\eqref{CartesianComponents_PTAI_z}, after lengthy but standard calculations we have
\begin{align}
\label{6D_derparE_t}
\derpar{e_{\alpha}}{t_{\beta}}=- \oj \chi \dfrac{\mathrm{e}^{-\oj kr}}{r} \cdot \begin{cases}
      1-r_{\alpha}^2& \text{per $\beta=\alpha$}, \\
      -r_{\alpha}r_{\beta}& \text{per $\beta \ne \alpha$}
\end{cases}
\end{align}
with $\alpha,\beta \in \{x,y,z\}$ and $r_{\alpha}$ being defined in \eqref{rcomp},
\begin{subequations}
\label{6D_derParEx}
\begin{align}
\label{6D_Ex_xC}\notag
\derpar{e_x}{x_C}=\dfrac{-\oj \chi \mathrm{e}^{-\oj kr}}{r^4}\left\{\vphantom{\frac{1}{1}}\oj k \bar{x} \left[t_x (\bar{y}^2+\bar{z}^2)-\bar{x}(t_y \bar{y}+t_z \bar{z})\right]\right.\\ \left.+\dfrac{3 t_x \bar{x}(\bar{y}^2+\bar{z}^2)-(t_y \bar{y}+t_z \bar{z}) (2 \bar{x}^2-\bar{y}^2-\bar{z}^2)}{r} \right\}
\end{align}
\begin{align}
\label{6D_Ex_yC}\notag
\derpar{e_x}{y_C}=\dfrac{-\oj \chi \mathrm{e}^{-\oj kr}}{r^4} \left\{\vphantom{\frac{1}{1}}\oj k \bar{y} \left[t_x (\bar{y}^2+\bar{z}^2)-\bar{x}(t_y \bar{y}+t_z \bar{z})\right]\right.\\ \left.-\dfrac{t_x \bar{y} (2 \bar{x}^2-\bar{y}^2-\bar{z}^2) - t_y \bar{x} (\bar{x}^2 -2 \bar{y}^2 +\bar{z}^2)+3 t_z \bar{x} \bar{y} \bar{z}}{r}\right\}
\end{align}
\begin{align}
\label{6D_Ex_zC}\notag
\derpar{e_x}{z_C}=\dfrac{-\oj \chi \mathrm{e}^{-\oj kr}}{r^4} \left\{\vphantom{\frac{1}{1}}\oj k \bar{z} \left[t_x (\bar{y}^2+\bar{z}^2)-\bar{x}(t_y \bar{y}+t_z \bar{z})\right]\right.\\ \left.-\dfrac{t_x \bar{z} (2 \bar{x}^2-\bar{y}^2-\bar{z}^2) +3 t_y \bar{x} \bar{y} \bar{z}- t_z \bar{x} (\bar{x}^2 + \bar{y}^2 -2\bar{z}^2)}{r}\right\}
\end{align}
\end{subequations}

\begin{subequations}
\label{6D_derParEy}
\begin{align}
\label{6D_Ey_xC}\notag
\derpar{e_y}{x_C}=\dfrac{-\oj \chi \mathrm{e}^{-\oj kr}}{r^4} \left\{\vphantom{\frac{1}{1}}\oj k \bar{x} \left[t_y (\bar{x}^2+\bar{z}^2)-\bar{y}(t_x \bar{x}+t_z \bar{z})\right]\right.\\ \left.-\dfrac{t_y \bar{x} (2 \bar{y}^2-\bar{x}^2-\bar{z}^2) - t_x \bar{y} (\bar{y}^2 -2 \bar{x}^2 +\bar{z}^2)+3 t_z \bar{x} \bar{y} \bar{z}}{r}\right\}
\end{align}
\begin{align}
\label{6D_Ey_yC}\notag
\derpar{e_y}{y_C}=\dfrac{-\oj \chi \mathrm{e}^{-\oj kr}}{r^4}\left\{\vphantom{\frac{1}{1}}\oj k \bar{y} \left[t_y (\bar{x}^2+\bar{z}^2)-\bar{y}(t_x \bar{x}+t_z \bar{z})\right]\right.\\ \left.+\dfrac{3 t_y \bar{y}(\bar{x}^2+\bar{z}^2)-(t_x \bar{x}+t_z \bar{z}) (2 \bar{y}^2-\bar{x}^2-\bar{z}^2)}{r} \right\}
\end{align}
\begin{align}
\label{6D_Ey_zC}\notag
\derpar{e_y}{z_C}=\dfrac{-\oj \chi \mathrm{e}^{-\oj kr}}{r^4} \left\{\vphantom{\frac{1}{1}}\oj k \bar{z} \left[t_y (\bar{x}^2+\bar{z}^2)-\bar{y}(t_x \bar{x}+t_z \bar{z})\right]\right.\\ \left.-\dfrac{t_y \bar{z} (2 \bar{y}^2-\bar{x}^2-\bar{z}^2) +3 t_x \bar{x} \bar{y} \bar{z}- t_z \bar{y} (\bar{x}^2 + \bar{y}^2 -2\bar{z}^2)}{r}\right\}
\end{align}
\end{subequations}

\begin{subequations}
\label{6D_derParEz}
\begin{align}
\label{6D_Ez_xC}\notag
\derpar{e_z}{x_C}=\dfrac{-\oj \chi \mathrm{e}^{-\oj kr}}{r^4} \left\{\vphantom{\frac{1}{1}}\oj k \bar{x} \left[t_z (\bar{x}^2+\bar{y}^2)-\bar{z}(t_x \bar{x}+t_y \bar{y})\right]\right.\\ \left.- \dfrac{t_z \bar{x} (2 \bar{z}^2-\bar{x}^2-\bar{y}^2) +3 t_y \bar{x} \bar{y} \bar{z}- t_x \bar{z} (\bar{z}^2 + \bar{y}^2 -2\bar{x}^2)}{r}\right\}
\end{align}
\begin{align}
\label{6D_Ez_yC}\notag
\derpar{e_z}{y_C}=\dfrac{-\oj \chi \mathrm{e}^{-\oj kr}}{r^4} \left\{\vphantom{\frac{1}{1}}\oj k \bar{y} \left[t_z (\bar{x}^2+\bar{y}^2)-\bar{z}(t_x \bar{x}+t_y \bar{y})\right] \right.\\ \left. - \dfrac{t_z \bar{y} (2 \bar{z}^2-\bar{x}^2-\bar{y}^2) - t_y \bar{z} (\bar{z}^2 -2 \bar{y}^2 +\bar{x}^2)+3 t_x \bar{x} \bar{y} \bar{z}}{r}\right\}
\end{align}
\begin{align}
\label{6D_Ez_zC}\notag
\derpar{e_z}{z_C}=\dfrac{-\oj \chi \mathrm{e}^{-\oj kr}}{r^4}\left\{\vphantom{\frac{1}{1}} \oj k \bar{z} \left[t_z (\bar{x}^2+\bar{y}^2)-\bar{z}(t_x \bar{x}+t_y \bar{y})\right] \right.\\ \left. + \dfrac{3 t_z \bar{z}(\bar{x}^2+\bar{y}^2)-(t_x \bar{x}+t_y \bar{y}) (2 \bar{z}^2-\bar{y}^2-\bar{x}^2)}{r} \right\}.
\end{align}
\end{subequations}
For convenience, we have set $\bar{x}=-x_C$, $\bar{y}=y-y_C$ e $\bar{z}=z-z_C$.

 %!TEX root = main.tex
\section*{Appendix B}
In this appendix, we prove the results of Proposition~\ref{CPL_VD}. We start by observing that, under Assumption~\ref{ass:cpl_vertical}, the derivatives needed for the computation of FIM in \eqref{matF} are obtained from the equations in Appendix A by simply setting $y_C=z_C=0$, $t_x=t_y=0$ and $t_z=1$. 

We first computate the elements of $\mathbf{F}_{cc}$ given by
{\begin{equation}
\label{Fcc}
\begin{split}
[\mathbf{F}_{cc}]_{mn}=\dfrac{2}{\sigma^2}\re\left\{\int\limits_{-\frac{L}{2}}^{\frac{L}{2}} \int\limits_{-\frac{L}{2}}^{\frac{L}{2}} f^{(cc)}_{mn}(y,z) dy dz\right\}
\end{split}
\end{equation}}
where 
\begin{equation}
\label{f_cc_mn}
f^{(cc)}_{mn}(y,z)=\derpar{e_x(\vect{r})}{a_m} \derpar{e_x^\ast(\vect{r})}{a_n}+\derpar{e_y(\vect{r})}{a_m} \derpar{e_y^\ast(\vect{r})}{a_n} +\derpar{e_z(\vect{r})}{a_m} \derpar{e_z^\ast(\vect{r})}{a_n} \phantom{\Bigg]}
\end{equation}
with $a_1=x_C$, $a_2=y_C$ and $a_3=z_C$. In this case, it can be shown that, for $m \ne n$, $f^{(cc)}_{mn}(y,z)$ in \eqref{f_cc_mn} is an odd function of $y$ and $z$, i.e.,
$$f^{(cc)}_{mn}(y,z)=-f^{(cc)}_{mn}(-y,z)=-f^{(cc)}_{mn}(y,-z).$$
As a result, due to the symmetry of the integration domain, the off-diagonal elements of $\mathbf{F}_{cc}$ are zero, meaning that $\mathbf{F}_{cc}$ is a diagonal matrix. As for the diagonal elements, standard but lengthy calculations show that they can be written as in \eqref{eq:Fcc11_A4}--\eqref{eq:Fcc33_A4}
where
\begingroup
\allowdisplaybreaks
\begin{align}
\label{I1}
\mathscr{I}_1 \triangleq & \dfrac{\rho}{(4+\rho^2)}\left[\dfrac{(14+3\rho^2)}{\sqrt{4+\rho^2}} \arctan \dfrac{\rho}{\sqrt{4+\rho^2}}+\dfrac{\rho}{(2+\rho^2)}\right]\\
\label{I2}
\mathscr{I}_2 \triangleq &\int\limits_{-\rho/2}^{\rho/2}\int\limits_{-\rho/2}^{\rho/2}\dfrac{1+u^2v^2+v^4}{(1+u^2+v^2)^4}dudv\\
\label{I3}
\mathscr{I}_3 \triangleq & \int\limits_{-\rho/2}^{\rho/2} \int\limits_{-\rho/2}^{\rho/2} \dfrac{u^2(1+u^2)}{(1+u^2+v^2)^3}dudv\\
\label{I4}
\mathscr{I}_4\triangleq & \int\limits_{-\rho/2}^{\rho/2}\int\limits_{-\rho/2}^{\rho/2} \dfrac{u^2(1+u^2) +v^2(1+v^2)-u^2v^2}{(1+u^2+v^2)^4}dudv\\
\label{I5}
\mathscr{I}_5\triangleq &\int\limits_{-\rho/2}^{\rho/2} \int\limits_{-\rho/2}^{\rho/2} \dfrac{v^2(1+u^2)}{(1+u^2+v^2)^3}dudv\\
\label{I6}
\begin{split}
\mathscr{I}_6\triangleq &\dfrac{\rho}{2(4+\rho^{2})^{2}}\Bigg[\dfrac{(9 \rho^4+76 \rho^2 +136)}{\sqrt{4+\rho^2}}\arctan\frac{\rho}{\sqrt{4+\rho^2}} \\ &+\dfrac{\rho(3 \rho^4+4 \rho^2 -8)}{(2+\rho^2)^2} \Bigg]
\end{split}
\end{align}
\endgroup
with $\rho=L/x_C$. 
Now, we consider the computation of $\mathbf{F}_{tt}$ with elements
\begin{equation}
\label{Ftt}
\begin{split}
[\mathbf{F}_{tt}]_{mn}=\dfrac{2}{\sigma^2}\re\left\{\int\limits_{-\frac{L}{2}}^{\frac{L}{2}} \int\limits_{-\frac{L}{2}}^{\frac{L}{2}} f^{(tt)}_{mn}(y,z) dy dz\right\}
\end{split}
\end{equation}
where 
\begin{equation}
\label{f_tt_mn}
f^{(tt)}_{mn}(y,z)=\derpar{e_x(\vect{r})}{b_m} \derpar{e_x^\ast(\vect{r})}{b_n}+\derpar{e_y(\vect{r})}{b_m} \derpar{e_y^\ast(\vect{r})}{b_n} +\derpar{e_z(\vect{r})}{b_m} \derpar{e_z^\ast(\vect{r})}{b_n} \phantom{\Bigg]}
\end{equation}
with $b_1=t_x$, $b_2=t_y$ and $b_3=t_z$.  Again, it can easily be shown that, for $m \ne n$, $f^{(tt)}_{mn}(y,z)$ in \eqref{f_tt_mn} is an odd function of $y$ or $z$, and hence the off-diagonal elements of $\mathbf{F}_{tt}$ are zero, meaning that $\mathbf{F}_{tt}$ is diagonal. The diagonal elements can be written as in \eqref{Ftt_11}--\eqref{Ftt_22}
where
\begin{align}
\label{I7}
\mathscr{I}_7 \triangleq &\int\limits_{-\rho/2}^{\rho/2}\int\limits_{-\rho/2}^{\rho/2}\dfrac{u^2+v^2}{(1+u^2+v^2)^2}dudv\\
\label{I8}
\mathscr{I}_8 \triangleq & \int\limits_{-\rho/2}^{\rho/2} \int\limits_{-\rho/2}^{\rho/2} \dfrac{1+u^2}{(1+u^2+v^2)^2}dudv.
\end{align}
Consider now the matrix $\mathbf{F}_{tc}$ with elements given by
{\begin{equation}
\label{Ftc}
\begin{split}
[\mathbf{F}_{tc}]_{mn}=\dfrac{2}{\sigma^2}\re\left\{\int\limits_{-\frac{L}{2}}^{\frac{L}{2}} \int\limits_{-\frac{L}{2}}^{\frac{L}{2}} f^{(tc)}_{mn}(y,z) dy dz \right\}
\end{split}
\end{equation}}
where 
\begin{equation}
\label{f_tc_mn}
f^{(tc)}_{mn}(y,z)=\derpar{e_x(\vect{r})}{b_m} \derpar{e_x^\ast(\vect{r})}{a_n}+\derpar{e_y(\vect{r})}{b_m} \derpar{e_y^\ast(\vect{r})}{a_n} +\derpar{e_z(\vect{r})}{b_m} \derpar{e_z^\ast(\vect{r})}{a_n} \phantom{\Bigg]}
\end{equation}
and $a_n,b_m$ have the same meaning as before. It can be shown that $f^{(tc)}_{13}(y,z)$ and $f^{(tc)}_{31}(y,z)$ are even functions of $y$ and $z$ whereas the others are odd functions of $y$ or $z$. {As a consequence, the only non-zero entries of $\mathbf{F}_{tc}$ are  $[\mathbf{F}_{tc}]_{13}$ and $[\mathbf{F}_{tc}]_{31}$ given by \eqref{Ftc13} and \eqref{Ftc31}, respectively, where
\begin{align}
\label{I9}
\mathscr{I}_9&\triangleq  \dfrac{2 \rho}{4+\rho^{2}} \left[\dfrac{2+\rho^{2}}{\sqrt{4+\rho^{2}}}\arctan\dfrac{\rho}{\sqrt{4+\rho^2}} -\dfrac{\rho}{2+\rho^2}\right]
\end{align}
\begin{align}
\label{I10}
\mathscr{I}_{10}&\triangleq \int\limits_{-\rho/2}^{\rho/2}\int\limits_{-\rho/2}^{\rho/2} \dfrac{(1+2u^2)}{(1+u^2+v^2)^{4}}dudv.
\end{align}}

 %!TEX root = main.tex
\section*{Appendix C}
The proof of Proposition 2 is given. {We start by deriving \eqref{eq:crb_ori_x_app1}. As a first step, we show that $k^2 \mathscr{I}_1 \gg x_C^{-2}\mathscr{I}_2$, for $x_C \gg \lambda$. To this end, observe that \begin{equation}
\label{I1_new}
\mathscr{I}_1=\int\limits_{-\rho/2}^{\rho/2}\int\limits_{-\rho/2}^{\rho/2} \dfrac{1+v^2}{(1+u^2+v^2)^3}dudv.
\end{equation}
Comparing \eqref{I1_new} and \eqref{I2} yields $\mathscr{I}_1 > \mathscr{I}_2$ since
\begin{equation}
\dfrac{1+v^2}{(1+u^2+v^2)^3} \ge \dfrac{1+u^2v^2+v^4}{(1+u^2+v^2)^4}.
\end{equation}
Since $0 \le \mathscr{I}_{2} - \mathscr{I}_8^{-1} \mathscr{I}_{10}^{2} < \mathscr{I}_{2}$ (notice that $\mathscr{I}_2 - \mathscr{I}_8^{-1} \mathscr{I}_{10}^{2}$ must necessarily be non-negative and $\mathscr{I}_8^{-1} \mathscr{I}_{10}^{2} >0$), for $x_C \gg \lambda$ we have
\begin{equation}
\label{I1_gg_I2}
4 \pi^2 \lambda^{-2} \mathscr{I}_1 \gg x_C^{-2}\mathscr{I}_2 > x_C^{-2}(\mathscr{I}_{2} - \mathscr{I}_8^{-1} \mathscr{I}_{10}^{2}).
\end{equation} 
The approximation in \eqref{eq:crb_ori_x_app1} derives from \eqref{I1_gg_I2} immediately.
}

We now derive \eqref{eq:crb_ori_y_app1} by showing that $k^2 \mathscr{I}_3 \gg x_C^{-2}\mathscr{I}_4$. Observe that
\begin{align}
\label{}
0 < \mathscr{I}_4 < \int\limits_{-\rho/2}^{\rho/2} \int\limits_{-\rho/2}^{\rho/2} \dfrac{2u^2(1+u^2)}{(1+u^2+v^2)^4} dudv < 2 \, \mathscr{I}_3
\end{align}
from which we easily obtain
\begin{equation}
\label{I3ggI4}
4 \pi^2 \lambda^{-2} \mathscr{I}_3 \gg x_C^{-2}\mathscr{I}_4
\end{equation}
for $x_C \gg \lambda$. The approximation in \eqref{eq:crb_ori_y_app1} follows from \eqref{I3ggI4} immediately.

 %!TEX root = main.tex
\section*{Appendix D}
\label{app:calc6D}
{The analysis of the CRBs in the asymptotic regime $\rho \to \infty$ requires the computation of the following limits:
\begin{enumerate}
  \item $\lim\limits_{\rho \to \infty} \mathscr{I}_1$
  \item $\lim\limits_{\rho \to \infty} \mathscr{I}_3$
  \item $\lim\limits_{\rho \to \infty} \left(k^2 \mathscr{I}_5+  x_C^{-2}\mathscr{I}_6\right)$
  \item $\lim\limits_{\rho \to \infty} \left[k^2 \mathscr{I}_5+  x_C^{-2}\left(\mathscr{I}_6-\mathscr{I}_7^{-1}\mathscr{I}_{9}^2\right)\right]$. 
\end{enumerate}
 We start by evaluating the first. From \eqref{I1}, we have
 \begin{equation}
\label{limI1}
\lim\limits_{\rho \to \infty} \mathscr{I}_1=3 \pi /4.
\end{equation}}

Consider now $\mathscr{I}_3$. Due to the non-negativity of the function $u^2(1+u^2)/(1+u^2+v^2)^3$, we have
\begin{equation}
\label{I3LUB}
\int\limits_{\mathcal{C}^-} \dfrac{u^2(1+u^2)}{(1+u^2+v^2)^3}dudv < \mathscr{I}_3 < \int\limits_{\mathcal{C}^+} \dfrac{u^2(1+u^2)}{(1+u^2+v^2)^3}dudv
\end{equation}
The two integrals can be computed in closed form. In particular, we obtain
\begin{equation}
\label{I3LB}
\int\limits_{\mathcal{C}^-} \dfrac{u^2(1+u^2)}{(1+u^2+v^2)^3}dudv=\dfrac{3\pi}{8} \ln (1+ \rho^2) - \dfrac{\pi}{16} \dfrac{\rho^2(5 \rho^2 +6)}{(1+\rho^2)^2}
\end{equation}
\begin{equation}
\label{I3UB}
\int\limits_{\mathcal{C}^+} \dfrac{u^2(1+u^2)}{(1+u^2+v^2)^3}dudv=\dfrac{3\pi}{8} \ln (1+ 2 \rho^2) - \dfrac{\pi}{4} \dfrac{\rho^2(5 \rho^2 +3)}{(1+2\rho^2)^2}.
\end{equation}
Taking \eqref{I3LUB} and \eqref{I3LB}--\eqref{I3UB} into account yields
\begin{equation}
\label{limI3}
\mathscr{I}_3 \sim \dfrac{3\pi}{4} \ln \rho \qquad \textrm{as } \rho \to \infty
\end{equation}
from which \eqref{eq:CRBYtoINF} is derived straightforwardly.

{Now, we focus on the limits
\begin{equation}
\lim\limits_{\rho \to \infty} \left(k^2\mathscr{I}_5 +  x_C^{-2}\mathscr{I}_6\right)
\label{limLast}
\end{equation}
and
\begin{equation}
\label{lim4}
\lim\limits_{\rho \to \infty} \left[k^2 \mathscr{I}_5+  x_C^{-2}\left(\mathscr{I}_6-\mathscr{I}_7^{-1}\mathscr{I}_{9}^2\right)\right].
\end{equation}
By using similar arguments to those for $\mathscr{I}_3$, it can be shown that 
\begin{equation}
\label{limI5}
\mathscr{I}_5 \sim \dfrac{\pi}{4} \ln \rho \qquad \textrm{as } \rho \to \infty.
\end{equation}
and
\begin{equation}
\label{limI7}
\lim\limits_{\rho \to \infty} \mathscr{I}_7 = \infty.
\end{equation}
As for $\mathscr{I}_6$ and $\mathscr{I}_9$, from \eqref{I6} and \eqref{I9} we have that
\begin{equation}
\label{limI6}
\lim\limits_{\rho \to \infty} \mathscr{I}_6=9 \pi /8.
\end{equation}
and
\begin{equation}
\label{limI9}
\lim\limits_{\rho \to \infty} \mathscr{I}_9=\pi /2.
\end{equation}
By taking \eqref{limI5}--\eqref{limI9} into account, \eqref{eq:CRBZtoINF} follows easily.}

\bibliographystyle{IEEEtran}
\bibliography{refs}

\end{document}